\documentclass[12pt]{article}
\usepackage{amscd,amsfonts,amssymb,amsmath,latexsym,array,hhline}

\usepackage{titlesec}
\titleformat{\subsection}[runin]
{\normalfont\bfseries}{\thesubsection}{1em}{} \mathsurround=1pt

\usepackage{color}

\renewcommand{\bar}{\overline}
\renewcommand{\tilde}{\widetilde}

\newcommand{\PP}{{\mathbb P}}
\newcommand{\EE}{{\mathbb E}}

\newcommand{\LL}{{\mathbb L}}

\newcommand{\FF}{{\mathbb F}}

\newcommand{\GG}{{\mathbb G}}
\newcommand{\NN}{{\mathbb N}}
\newcommand{\RR}{{\mathbb R}}

\renewcommand{\phi}{\varphi}
\renewcommand{\kappa}{\varkappa}

\newtheorem{theorem}{Theorem}[section]

\newtheorem{proposition}[theorem]{Proposition}

\newenvironment{proof}{{\vskip\baselineskip\noindent\textbf{Proof}}}%

 \numberwithin{equation}{section}

 \renewcommand{\section}[1]{\refstepcounter{section}
                            {\noindent\large\bf\thesection. #1}}

\makeatletter

\renewcommand{\section}{\@startsection{section}{1}{0pt}{30pt}{6pt}{\large\bf}}
\def\dot{\hspace{-16pt}.\hspace{-2pt} }

\renewcommand{\@makefnmark}{}
\renewcommand{\@cite}[2]{[{#1\if@tempswa ; #2\fi}]}

\makeatother

\begin{document}

\title{Perpetual American Standard and Lookback Options in Insider Models with Progressively Enlarged Filtrations}

\author{Pavel V. Gapeev\footnote{(Corresponding author) London School of Economics,
Department of Mathematics, Houghton Street, London WC2A 2AE, United
Kingdom; e-mail: p.v.gapeev{\char'100}lse.ac.uk} \and Libo Li\footnote{University of New South Wales, School of Mathematics and Statistics, Sydney NSW 2052, Australia; e-mail: libo.li{\char'100}unsw.edu.au}}

\date{}
\maketitle

\begin{abstract}
       We derive closed-form solutions to the optimal stopping problems related to
       the pricing of perpetual American standard and lookback put and call options
       in the extensions of the Black-Merton-Scholes model with progressively enlarged filtrations.
       More specifically, the information available to the insider is modelled
       by Brownian filtrations progressively enlarged with the times of either the global
       maximum or minimum of the underlying risky asset price over the infinite time interval, which is not a stopping time in the filtration generated by the underlying risky asset.
       We show that the optimal exercise times are the first times at which the asset price process
       reaches either lower or upper stochastic boundaries depending on the current values of its
       running maximum or minimum given the occurrence of times of either the global maximum or minimum,
       respectively.
       The proof is based on the reduction of the original problems into the necessarily
       three-dimensional optimal stopping problems and the equivalent free-boundary problems.
       We apply either the normal-reflection or the normal-entrance conditions as well as the smooth-fit
       conditions for the value functions to characterise the candidate boundaries as either the
       maximal or minimal solutions to the associated first-order nonlinear ordinary differential
       equations and the transcendental arithmetic equations, respectively.
\end{abstract}





\footnotetext{{\it Date:} \date{\today}}

\section{\dot Introduction}


     Let us consider a probability space $(\Omega, {\cal F}, \PP)$ with a standard Brownian motion
     $B=(B_t)_{t \ge 0}$ and define the process $X=(X_t)_{t \ge 0}$ by:
       \begin{equation}
       \label{X4}
       X_t = x \, \exp \Big( \big( r - \delta - {\sigma^2}/{2} \big) \, t + \sigma \, B_t \Big)
       \end{equation}
       which solves the stochastic differential equation:
       \begin{equation}
       \label{dX4}
       dX_t = (r-\delta) \, X_{t} \, dt + \sigma \, X_{t} \, dB_t \quad (X_0=x)
       \end{equation}
       where $x > 0$ is fixed, and $r > 0$, $\delta > 0$, and $\sigma > 0$ are some given constants.
       Assume that $X$ describes the price of a risky asset in a financial market,
       where $r$ is the riskless interest rate, $\delta$ is the dividend rate paid to the asset holders,
       and $\sigma$ is the volatility rate.
     We aim to present closed-form solutions to the discounted optimal stopping problems with the values:
     \begin{align}
     \label{VU5a1}
     V_i &= \sup_{\tau \in \GG^1} \EE \big[ e^{- r \tau} \, G_i(X_{\tau}, S_{\tau}) \big]
     \quad \text{and} \quad
     U_i = \sup_{\zeta \in \GG^2} \EE \big[ e^{- r \zeta} \, F_i(X_{\zeta}, Q_{\zeta}) \big]
     \end{align}
     where $G_1(x, s) = L_1 - x$, $G_2(x, s) = s - L_2 x$, $G_3(x, s) = s - L_3$, and
     $F_1(x, q) = x - K_1$, $F_2(x, q) = K_2 x - q$, $F_3(x, q) = K_3 - q$,
     for some deterministic constants $L_i, K_i > 0$, for $i = 1, 2, 3$.
     The linear functions $G_i(x, s)$ and $F_i(x, q)$, for $i = 1, 2, 3$, represent the
     payoffs of {\it standard} and {\it lookback} options with {\it floating} and {\it fixed}
     strikes, respectively, which are widely used in financial practice.
     Here, the processes $S=(S_t)_{t \ge 0}$ and $Q=(Q_t)_{t \ge 0}$
     are the {\it running maximum} and {\it minimum} of $X$ defined by:
       \begin{equation}
       \label{SQ4}
       S_t = s \vee \Big( \max_{0 \le u \le t} X_u \Big) \equiv \max \Big\{ s, \max_{0 \le u \le t} X_u \Big\}
       \quad \text{and} \quad
       Q_t = q \wedge \Big( \min_{0 \le u \le t} X_u \Big) \equiv \min \Big\{ q, \min_{0 \le u \le t} X_u \Big\}
       \end{equation}
       for some arbitrary $0 < x \le s$ and $0 < q \le x$, respectively.

     Suppose that the suprema in (\ref{VU5a1}) are taken over all stopping times
     $\tau$ and $\zeta$ with respect to the filtrations $\GG^k = ({\cal G}^k_t)_{t \ge 0}$,
     for $k = 1, 2$, which coincide with the (right-continuous) reference filtration
     $\FF = ({\cal F}_t)_{t \ge 0}$ progressively enlarged by the random times $\theta$ and $\eta$,
     which are given by:
       \begin{equation}
       \label{theta}
       \theta = \sup \{ t \ge 0 \, | \, X_t = S_t \}
       \quad \text{and} \quad
       \eta = \sup \{ t \ge 0 \, | \, X_t = Q_t \}
       \end{equation}
       respectively. More precisely, the filtrations $\GG^k$, for $k = 1, 2$,
       are defined as the the right-continuous versions of $\GG^{*,k}$, for $k = 1, 2$, given by:
\begin{equation}
\label{gg}
\GG^{*,1}_t = \bigcap_{u > t} \big( {\cal F}_u \vee \sigma (\theta \wedge u) \big)
\quad \text{and} \quad
\GG^{*,2}_t = \bigcap_{u > t} \big( {\cal F}_u \vee \sigma (\eta \wedge u) \big)
\end{equation}
for all $t \ge 0$.
       Note that the random times $\theta$ and $\eta$ from (\ref{theta})
       are not stopping times with respect to the natural filtration $\FF$
       of the process $X$, but they are honest times with respect to $\FF$
       in the sense of Barlow \cite{B} and Nikeghbali and Yor \cite{NY}.
       In this view, the values $V_i$ and $U_i$, for $i = 1, 2, 3$, in (\ref{VU5a1}) can be
       interpreted as the rational (or no-arbitrage) prices of the perpetual American options
       in the appropriate extension of the Black-Merton-Scholes model with enlarged information
       flows (see, e.g. \cite[Chapter~VII, Section~3g]{FM}).
       Some extensive overviews of the perpetual American options in diffusion models of financial
       markets and other related results in the area are provided in Shiryaev \cite[Chapter~VIII; Section~2a]{FM},
       Peskir and Shiryaev \cite[Chapter~VII; Section~25]{PSbook}, and Detemple \cite{Det} among others.

This work investigates the value of American options within the context of models with insider information. Our objective is to quantify the additional value to the holder (the informed agent) when they have access to the timing of the global maximum or minimum of the price of the underlying asset. We assume that the holder is permitted to exercise the option at these specific times, or more generally, stopping times in the progressive enlargements $\GG^k$, for $k = 1, 2$, of the market filtration $\FF$ with the timing of the global maximum or minimum of the asset price $X$. This expands on our previous work \cite{GL1}-\cite{GL2} (and \cite{GLW}), where the standard and lookback options could only be exercised at $\FF$-stopping times. In this paper, we also stress that we do not claim that the informed agent (or some third parties communicating with them) know the future dynamics of the risky asset prices in advance but rather they are insiders who can either influence or determine the timing at which the price processes of the underlying risky assets reach their global peaks or bottoms. We emphasise that utilising the last passage times is a natural approach to modeling insider problems. Our assumptions are consistent with the observations of Associate Professor Daniel Taylor from the Wharton School, as noted in an NBC article\footnote{https://www.nbcnews.com/business/business-news/peloton-insiders-sold-nearly-500-million-stock-big-drop-rcna12741} on insider trading: {\it ``One of the most well-accepted facts from decades of research on insider trading is that corporate insiders buy near bottoms and sell near peaks."} Furthermore, it has been established that arbitrage opportunities exist in models involving last passage times, as discussed by Fontana, Jeanblanc and Song \cite{FJS} and Zwierz \cite{Zwierz}.


       The perpetual American standard and lookback option pricing problems with the payoffs
       $G_i(x, s)$ and $F_i(x, q)$, for $i = 1, 2, 3$, defined after (\ref{VU5a1}) for the
       exercise times under the Brownian reference filtration $\FF = ({\cal F}_t)_{t \ge 0}$
       generated by the process $X$ only were considered in Shepp and Shiryaev \cite{SS1}-\cite{SS2}
       (see also \cite[Chapter~VIII, Section~2]{FM} and \cite{PSbook}),
       Pedersen \cite{Jesper}, Guo and Shepp \cite{GuoShepp}, and, more recently,
       in Rodosthenous and Zervos \cite{RodZer}, Gapeev, Kort, and Lavrutich \cite{GKL},
       Gapeev et al. \cite{GKLT22}, and Gapeev \cite{Gapmax24}, respectively.
       The same problems in the models based on the geometric jump-diffusion processes
       and L\'evy processes having spectrally negative or positive jumps were studied by
       Asmussen, Avram, and Pistorius \cite{AAP}, Avram, Kyprianou, and Pistorius \cite{AKP},
       Gapeev \cite{Gap}, Ott \cite{Ott}, Kyprianou and Ott \cite{KyprOtt} among others.
       Other versions of the problems in the models given by (\ref{X4})-(\ref{dX4}) and
       (\ref{SQ4}) under the Brownian reference filtration $\FF$ with the same payoff
       functions but with the processes $(X, S)$ and $(X, Q)$ stopped at the random times
       $\theta$ and $\eta$ from (\ref{theta}) were considered in our previous papers
       \cite{GL1}-\cite{GL2} (see also \cite{GLW}).
       Although the random times $\theta$ and $\eta$, which are the last hitting times for
       the underlying risky asset price process $X$ of its running maximum $S$ or minimum $Q$
       are not stopping times with respect to $\FF$, the problems in \cite{GL1}-\cite{GL2}
       were solved as the associated optimal stopping problems
       for the two-dimensional processes $(X, S)$ and $(X, Q)$, respectively.
       In contrast to the similar problems studied in the literature, in this paper, we study
       the perpetual American standard and lookback options in the models in which the holders
       can also progressively observe the random times at which the underlying risky assets,
       which may represent the shares of either their own firms or other companies
       to the decisions of the board of which they have access and potentially influence them,
       attain their global maxima or mimima. This corresponds to the case in which the firms issuing
       the asset announce the appropriate executive board decisions implying the properties that the
       associated firm value processes achieve their global maxima or minima over the infinite time interval,
       respectively.

Investment problems in models of financial markets involving insider information are typically studied in the initial enlargement of filtrations, by focusing primarily on the insider utility or portfolio maximisation. It is usually assumed that the informed agent knows the terminal values of the underlying risky asset, which satisfy Jacod's density hypothesis. Notable contributions in this area include Amendinger, Imkeller, and Schweizer \cite{AIS}, Grorud and Pontier \cite{GP98}, and Ankirchner, Dereich, and Imkeller \cite{ADI2006}. Optimal stopping problems involving initial enlargement of filtration were also studied by Esmaeeli and Imkeller \cite{EI}, by means of a reflected BSDE approach. Dumitrescu, Quenez, and Sulem \cite{DQS} as well as Grigorova, Quenez, and Sulem \cite{GQS} studied American options in the context of progressively enlarged filtrations using reflected BSDEs, however, both \cite{DQS} and \cite{GQS} assume that the random times exhibit an intensity process and satisfy hypothesis $(H)$. In the case of game options, we refer to Bielecki, Cr\'epey, Jeanblanc and Rutkowski \cite{BCJR}-\cite{BCJR1} where the random time is assumed to be a pseudo-stopping time, see Nikeghbali and Yor \cite{NY1}. Recently, D'Auria, Di Nunno and Salmer\'on \cite{SDD} explored portfolio optimisation problems for informed agents with access to last passage times.

In the present paper, unlike the aforementioned works, we are faced additional mathematical challenges since, in general, last passage times are not pseudo-stopping times (cf., e.g. Aksamit and Li \cite[Proposition 1]{AL}) and do not satisfy the intensity hypothesis, hypothesis $(H)$ or the density hypothesis (cf., e.g. Aksamit and Jeanblanc \cite{aj}). To tackle the problem at hand, we embed the initial problems of (\ref{VU5a1}) into the associated optimal stopping problems of (\ref{VU5c1}) below for the three-dimensional continuous Markov processes having the underlying
       risky asset price $X$ and its running maximum $S$ or minimum $Q$ together with the so-called
       default processes $\Xi^1 \equiv I(\theta \le \cdot)$ and $\Xi^2 \equiv I(\eta \le \cdot)$
       as their state space components, respectively, where $I(\cdot)$ denotes the indicator function.
       The resulting problems turn out to be necessarily three-dimensional in the sense that they cannot
       (generally) be reduced to optimal stopping problems for (time-homogeneous strong) Markov processes of lower dimensions.
       The new analytic feature observed by solving these problems is that the normal-reflection
       conditions hold for the value functions at the edges of the two-dimensional state
       spaces of the processes $(X, S)$ and $(X, Q)$, when $\Xi^k = 0$, for $k = 1, 2$,
       while the normal-entrance conditions hold instead, when $\Xi^k = 1$, for $k = 1, 2$.
       This matter is explained by the facts that the processes $S$ and $Q$ change their
       values, when $X = S$ and $X = Q$ before the occurrence of the random times $\theta$
       and $\eta$, which are progressively observable in $\GG^k$, for $k = 1, 2$, but do not change
       their values after the occurrence of $\theta$ and $\eta$, respectively.
       The latter facts also imply the properties that the optimal exercise boundaries
       for the process $X$, which depend on the current values of $S$ and $Q$, remain
       constant after the occurrence of $\theta$ and $\eta$, respectively.
In particular, in the cases of floating-strike lookback option payoffs
$G_2(x, s) \equiv s - L_2 x$ and $F_2(x, q) \equiv q - K_2 x$, the optimal
stopping boundaries form straight lines, when $\Xi^1 = 0$ and $\Xi^2 = 0$, respectively.
These properties can be explained by the facts that the optimal
stopping problems of (\ref{VU5c1}) are reduce to the ones for
the processes $S/X$ and $Q/X$ as in \cite{SS2} and \cite{Gapmax24},
under either $\Xi^1 = 0$ or $\Xi^2 = 0$, respectively.

The rest of the paper is organised as follows. In Section~2, we embed the original option pricing
problems of (\ref{VU5a1}) into the optimal stopping problems of (\ref{VU5c1}) for the two-dimensional continuous
Markov processes $(X, S, \Xi^1 \equiv I(\theta \le \cdot))$ and $(X, Q, \Xi^2 \equiv I(\eta \le \cdot))$
given by (\ref{X4a})-(\ref{X4b}) and (\ref{X4b})-(\ref{dX4b}) with (\ref{SQ4}) and (\ref{theta}).
It is shown that the optimal stopping times $\tau^*_i$ and $\zeta^*_i$ are the first times at which
the process $X$ reaches either lower or upper boundaries $a^*_{i,\Xi^1}(S)$ or $b^*_{i,\Xi^2}(Q)$
depending on the current values of the processes $S$ or $Q$ and either $\Xi^1$ or $\Xi^2$,
for $i = 1, 2, 3$, respectively.
In Section~3, we derive closed-form expressions for the associated value functions $V^*_{i,j}(x, s)$
and $U^*_{i,j}(x, q)$, for $i = 1, 2, 3$ and $j = 0, 1$, as solutions to the equivalent free-boundary problems. We then apply the modified normal-reflection conditions at the edges of the three-dimensional state spaces for $(X, S, \Xi^1)$ or $(X, Q, \Xi^2)$, in order to characterise the optimal stopping boundaries $a^*_{i,\Xi^1}(S)$ and $b^*_{i,\Xi^2}(Q)$, for $i = 1, 2, 3$, as either the maximal or minimal solutions to the resulting first-order nonlinear ordinary differential equations.
In Section~4, by using the change-of-variable formula with local time on surfaces from Peskir \cite{Pe1a}, we verify that the solutions of the free-boundary problems provide the solutions of the original optimal stopping problems. The main results of the paper are stated in Theorems~\ref{lem21} and \ref{thm41}.

    \section{\dot Preliminaries}

     In this section, we introduce the setting and notation of the two-dimensional optimal stopping problems which are related to the pricing of perpetual American standard and lookback put and call options under the progressively enlarged filtrations $\GG^k$, for $k = 1, 2$, described above and formulate the equivalent free-boundary problems.


 \subsection{The model with progressively enlarged filtrations.}
    In order to handle the expectations in (\ref{VU5a1}), let us now introduce
    the conditional survival processes $Z=(Z_t)_{t \ge 0}$ and $Y=(Y_t)_{t \ge 0}$ defined by
    $Z_t = \PP(\theta > t \, | \, {\cal F}_t)$ and $Y_t = \PP(\eta > t \, | \, {\cal F}_t)$,
    for all $t \ge 0$, respectively. Note that the processes $Z$ and $Y$ are called the Az\'ema
    supermartingales of the random times $\theta$ and $\eta$ (see, e.g. \cite[Section~1.2.1]{MY}).
    By using the fact that the global maximum of a Brownian motion with the drift coefficient $(r - \delta)/\sigma^2 - 1/2 < 0$ has an exponential distribution with the mean $1/(2(\delta - r)/\sigma^2 + 1)$, while the negative of the global minimum of a Brownian motion with the drift coefficient $(r - \delta)/\sigma^2 - 1/2 > 0$ has an exponential distribution with the mean $1/(2(r - \delta)/\sigma^2 - 1)$ (cf., e.g. \cite[Chapter~II, Exercise~3.12]{RY}), we have:
       \begin{equation}
       \label{Pi1b3}
       Z_t =
       \begin{cases}
       ({S_t}/{X_t})^{\alpha}, & \text{if} \;\; \alpha < 0 \\
       1, & \text{if} \;\; \alpha \ge 0
       \end{cases}
       \quad \text{and} \quad
       Y_t =
       \begin{cases}
       ({Q_t}/{X_t})^{\alpha}, & \text{if} \;\; \alpha > 0 \\
       1, & \text{if} \;\; \alpha \le 0
       \end{cases}
       \end{equation}
       for all $t \ge 0$, under $s = x$ and $q = x$,
       where we set $\alpha = 2 (r - \delta)/{\sigma^2} - 1$, respectively.
       The representations in (\ref{Pi1b3}) can also be obtained
	   from applying Doob's maximal equality (cf. \cite[Lemma~2.1 and Proposition~2.2]{NY})
	   to the process $X^{-\alpha} = (X^{-\alpha}_t)_{t \ge 0}$,
       which is a strictly positive continuous local martingale converging to zero at infinity.

       Moreover, it follows from standard applications of It\^o's formula (cf., e.g. \cite[Chapter~IV, Theorem~3.3]{RY})
       and the properties that the processes $S$ and $Q$ may change their values only when $X_t = S_t$ and $X_t = Q_t$, for $t \ge 0$,
       respectively, that the Az\'ema supermartingales $Z$ and $Y$ from (\ref{Pi1b3})
       admit the stochastic differentials:
       \begin{equation}
       \label{dZ3}
       dZ_t = - \alpha \, \bigg( \frac{S_t}{X_t} \bigg)^{\alpha} \, \sigma \, dB_t +
       \alpha \, I(X_t = S_t) \, \bigg( \frac{S_t}{X_t} \bigg)^{\alpha} \, \frac{dS_t}{S_t} =
       - \alpha \, \bigg( \frac{S_t}{X_t} \bigg)^{\alpha} \, \sigma \, dB_t + \alpha \, \frac{dS_t}{S_t}
       \end{equation}
       when $\alpha < 0$, and
       \begin{equation}
       \label{dY3}
       dY_t = - \alpha \, \bigg( \frac{Q_t}{X_t} \bigg)^{\alpha} \, \sigma \, dB_t +
       \alpha \, I(X_t = Q_t) \, \bigg( \frac{Q_t}{X_t} \bigg)^{\alpha} \, \frac{dQ_t}{Q_t} =
       - \alpha \, \bigg( \frac{Q_t}{X_t} \bigg)^{\alpha} \, \sigma \, dB_t + \alpha \, \frac{dQ_t}{Q_t}
       \end{equation}
       when $\alpha > 0$, respectively, where $I(\cdot)$ denotes the indicator function.

Let us now introduce the processes ${\tilde B}^1 = ({\tilde B}^1_t)_{t \ge 0}$
and ${\tilde B}^2 = ({\tilde B}^2_t)_{t \ge 0}$ defined by:
\begin{equation}
\label{Btil1}
{\tilde B}^1_t = B_t - \int_0^t \bigg( \alpha \sigma \, \frac{Z_u}{1-Z_u}
\, I(\theta \le u) - \alpha \sigma \, I(\theta > u) \bigg) \, du
\end{equation}
when $\alpha < 0$, and
\begin{equation}
\label{Btil2}
{\tilde B}^2_t = B_t - \int_0^t \bigg( \alpha \sigma \, \frac{Y_u}{1-Y_u}
\, I(\eta \le u) - \alpha \sigma \, I(\eta > u) \bigg) \, du
\end{equation}
when $\alpha > 0$, which are standard Brownian motions
under $(\GG^1, \PP)$ and $(\GG^2, \PP)$, for all $t \ge 0$, respectively, according to the results of \cite{JLC}.
In this case, using the expressions in (\ref{Pi1b3}), we obtain that the process $X$ from (\ref{X4})-(\ref{dX4})
admits the representations:
       \begin{equation}
       \label{X4a}
       X_t = x \, \exp \bigg( \int_0^t
       \Big( r - \delta + \Psi \big( (S_u/X_u)^{\alpha}, \Xi^1_u \big)
       - \frac{\sigma^2}{2} \Big) \, du + \sigma \, {\tilde B}^1_t \Big)
       \end{equation}
       and
       \begin{equation}
       \label{X4b}
       X_t = x \, \exp \bigg( \int_0^t
       \Big( r - \delta + \Phi \big( (Q_u/X_u)^{\alpha}, \Xi^2_u \big)
       - \frac{\sigma^2}{2} \Big) \, du + \sigma \, {\tilde B}^2_t \Big)
       \end{equation}
       which solve the stochastic differential equations:
       \begin{equation}
       \label{dX4a}
       dX_t = \big( r - \delta + \Psi \big( (S_t/X_t)^{\alpha}, \Xi^1_t \big) \big)
       \, X_{t} \, dt + \sigma \, X_{t} \, d{\tilde B}^1_t \quad (X_0=x)
       \end{equation}
       and
       \begin{equation}
       \label{dX4b}
       dX_t = \big( r - \delta + \Phi \big( (Q_t/X_t)^{\alpha}, \Xi^2_t \big) \big)
       \, X_{t} \, dt + \sigma \, X_{t} \, d{\tilde B}^2_t \quad (X_0=x)
       \end{equation}
       for all $t \ge 0$ and each $x > 0$ fixed. Here, we set:
       \begin{equation}
       \label{psi4}
       \Psi \big( (s/x)^{\alpha}, j \big) = \alpha \sigma \, \frac{(s/x)^{\alpha}}
       {1-(s/x)^{\alpha}} \, I(j = 1) - \alpha \sigma \, I(j = 0)
       \end{equation}
       for all $0 < x \le s$ and $j = 0, 1$, when $\alpha < 0$, and
       \begin{equation}
       \label{xi4}
       \Phi \big( (q/x)^{\alpha}, j \big) = \alpha \sigma \, \frac{(q/x)^{\alpha}}
       {1-(q/x)^{\alpha}} \, I(j = 1) - \alpha \sigma \, I(j = 0)
       \end{equation}
       for all $0 < q \le x$ and $j = 0, 1$, when $\alpha > 0$.


\subsection{The optimal stopping problems.}
Observe that the triples $(X_t, S_t, \Xi^1_t \equiv I(\theta \le t))_{t \ge 0}$ and
$(X_t, Q_t, \Xi^2_t \equiv I(\eta \le t))_{t \ge 0}$ form three-dimensional
(time-homogeneous strong) Markov processes under $(\GG^1, \PP)$ and $(\GG^2, \PP)$, respectively.
In this case, we can embed the initial pricing problems of (\ref{VU5c1}) with the linear
payoff functions $G_i(x, s)$ and $F_i(x, q)$ given after (\ref{VU5a1}) above into the
three-dimensional optimal stopping problems:
        \begin{equation}
        \label{VU5c1}
        V^*_{i,j}(x, s) = \sup_{\tau \in \GG^1}
        \EE_{x, s, j} \big[ e^{- r \tau} \, G_i(X_{\tau}, S_{\tau}) \big]
        \quad \text{and} \quad
        U^*_{i,j}(x, q) = \sup_{\zeta \in \GG^2}
        \EE_{x, q, j} \big[ e^{- r \zeta} \, F_i(X_{\zeta}, Q_{\zeta}) \big]
        \end{equation}
        where the suprema are taken over the stopping times $\tau$ and $\zeta$
       of the processes $(X, S, \Xi^1)$ and $(X, Q, \Xi^2)$, for every $i = 1, 2, 3$ and $j = 0, 1$, respectively.
       Here, $\EE_{x, s, j}$ and $\EE_{x, q, j}$ denote the expectations with respect to the probability
       measures $\PP_{x, s, j}$ and $\PP_{x, q, j}$ under which the three-dimensional Markov processes
       $(X, S, \Xi^1)$ and $(X, Q, \Xi^2)$ given by (\ref{X4a})-(\ref{dX4a}) and (\ref{X4b})-(\ref{dX4b})
       with (\ref{SQ4}) start at $(x, s, j) \in E_{1} \times \{ 0, 1 \} \equiv
       \{ (x, s) \in \RR^2 \, | \, 0 < x \le s \} \times \{ 0, 1 \}$
       and $(x, q, j) \in E_{2} \times \{ 0, 1 \} \equiv \{ (x, q) \in \RR^2 \, | \, 0 < q \le x \} \times \{ 0, 1 \}$,
       for $j = 0, 1$, respectively.
       We further obtain solutions to the optimal stopping problems in (\ref{VU5c1}) and verify below
       that the value functions $V^*_{i,j}(x, s)$ and $U^*_{i,j}(x, q)$, for $i = 1, 2, 3$ and $j = 0, 1$,
       are the solutions to the problems in (\ref{VU5c1}), and thus, give the solutions to the original
       problems in (\ref{VU5a1}), under $s = x$ and $q = x$ as well as $j = 0$, respectively.


        In order to handle the optimal stopping problems formulated in (\ref{VU5c1}),
        by means of standard applications of It\^o's formula,
        taking into account the facts that $\partial_{xx} G_i(x, s) = 0$
        and $\partial_{xx} F_i(x, q) = 0$,
        we obtain that the processes $(e^{-r t} G_i(X_t, S_t))_{t \ge 0}$
        and $(e^{-r t} F_i(X_t, Q_t))_{t \ge 0}$ admit the representations:
        \begin{align}
        \label{G51a}
        &e^{-r t} \, G_i(X_t, S_t) = G_i(x, s) \\
        \notag
        &+ \int_0^t e^{-r u} \, \Big( \partial_x G_i(X_u, S_u) \,
        \big( r - \delta + \Psi \big( (S_u/X_u)^{\alpha}, \Xi^1_u \big) \big) \, X_u
        - r \, G_i(X_u, S_u) \Big) \, I(X_u \neq S_u) \, du \\
        \notag
        &+ \int_0^t e^{-r u} \, \partial_s G_i(X_u, S_u) \, I(X_u = S_u) \, dS_u
        + \int_0^t e^{-r u} \, \partial_x G_i(X_u, S_u) \, I(X_u \neq S_u) \, \sigma X_u \, d{\tilde B}^1_u
        \end{align}
        when $\alpha < 0$, for each $0 < x \le s$, while
        \begin{align}
        \label{G51b}
        &e^{-r t} \, F_i(X_t, Q_t) = F_i(x, q) \\
        \notag
        &+ \int_0^t e^{-r u} \, \Big( \partial_x F_i(X_u, Q_u) \,
        \big( r - \delta + \Phi \big( ({Q_u}/{X_u})^{\alpha}, \Xi^2_u \big) \big) \, X_u
        - r \, F_i(X_u, Q_u) \Big) \, I(X_u \neq Q_u) \, du \\
        \notag
        &+ \int_0^t e^{-r u} \, \partial_q F_i(X_u, Q_u) \, I(X_u = Q_u) \, dQ_u
        + \int_0^t e^{-r u} \, \partial_x F_i(X_u, Q_u) \, I(X_u \neq Q_u) \, \sigma X_u \, d{\tilde B}^2_u
        \end{align}
        when $\alpha > 0$, for each $0 < q \le x$, for all $t \ge 0$, and every $i = 1, 2, 3$.
        Here, the last stochastic integral processes are continuous square integrable martingales,
        and thus, they are uniformly integrable martingales under the probability measure $\PP$,
        when $\alpha < 0$ and $\alpha > 0$, respectively.
       Note that the processes $S$ and $Q$ may change their values only at the times
       when $X_t = S_t$ and $X_t = Q_t$, for $t \ge 0$, respectively, and such times
       accumulated over the infinite horizon form the sets of the Lebesgue measure zero,
       so that the indicators in the expressions of (\ref{G51a})-(\ref{G51b}) can be ignored
       (cf. also Proof of Theorem~4.1 below for more explanations and references).
       Then, inserting $\tau$ and $\zeta$ in place of $t$ into (\ref{G51a}) and (\ref{G51b}),
       respectively, by means of Doob's optional sampling theorem
       (cf., e.g. \cite[Chapter~II, Theorem~3.2]{RY}), we get:
       \begin{align}
       \label{V3b5}
       &\EE_{x, s, j} \big[ e^{- r \tau} \, G_i(X_{\tau}, S_{\tau}) \big] = G_i(x, s) \\
       \notag
       &+ \EE_{x, s, j} \bigg[ \int_0^{\tau} e^{-r u} \, H_{i,1,\Xi^1_u}(X_u, S_u) \, du
       + \int_0^{\tau} e^{-r u} \, \partial_s G_i(X_u, S_u) \, dS_u \bigg]
       \end{align}
       when $\alpha < 0$, and
       \begin{align}
       \label{V3c5}
       &\EE_{x, q, j} \big[ e^{- r \zeta} \, F_i(X_{\zeta}, Q_{\zeta}) \big] = F_i(x, q) \\
       \notag
       &+ \EE_{x, q, j} \bigg[ \int_0^{\zeta} e^{-r u} \, H_{i,2,\Xi^2_u}(X_u, Q_u) \, du
       + \int_0^{\tau} e^{-r u} \, \partial_q F_i(Q_u, Q_u) \, dQ_u \bigg]
       \end{align}
       when $\alpha > 0$, for any stopping times $\tau$ and $\zeta$ as well as every $i = 1, 2, 3$. Here, we set:
        \begin{align}
        \label{H5b1}
        H_{i,1,j}(x, s) = \partial_x G_i(x, s) \,
        \big( r - \delta + \Psi \big( (s/x)^{\alpha}, j \big) \big) \, x - r \, G_i(x, s)
        \end{align}
        for all $0 < x \le s$, and
        \begin{align}
        \label{H5b2}
        H_{i,2,j}(x, q) = \partial_x F_i(x, q) \,
        \big( r - \delta + \Phi \big( (q/x)^{\alpha}, j \big) \big) \, x - r \, F_i(x, q)
        \end{align}
        for all $0 < q \le x$, as well as $i = 1, 2, 3$ and $j = 0, 1$, respectively.

	Observe from the expressions in (\ref{VU5c1}) and (\ref{V3b5}) with (\ref{H5b1})
    that if the process $(X, S, \Xi^1)$ starts at the point $(x, s, 1)$
	with $s \le \underline{s}_{3,0} = L_3$ at which $S$ attains its global maximum,
	then $\tau^*_3 = \infty$, while if it starts at $(x, s, 1)$ with
	$s > s_{3,0} = L_3$, then $\tau^*_3 = 0$.
    It is also seen from the expressions in (\ref{VU5c1}) and (\ref{V3c5}) with (\ref{H5b2})
    that if the process $(X, Q, \Xi^2)$ starts at the point $(x, q, 1)$
	with $q \ge \overline{q}_{3,0} = K_3$ at which $Q$ attains its global minimum,
	then $\zeta^*_3 = \infty$, while if it starts at $(x, q, 1)$ with
	$q < \overline{q}_{3,0} = K_3$, then $\zeta^*_3 = 0$.

\subsection{The structure of optimal exercise times.}
        Let us now determine the structure of the optimal stopping times
        at which the holders should exercise the contracts with the values
        in (\ref{VU5c1}). We summarise all the related arguments in the
        following assertion.

\begin{theorem} \label{lem21}
Let the processes $(X, S, \Xi^1 \equiv I(\theta \le \cdot))$ and
$(X, Q, \Xi^2 \equiv I(\eta \le \cdot))$ be given by (\ref{X4a})-(\ref{X4b}) and
(\ref{dX4a})-(\ref{dX4b}) with (\ref{SQ4})-(\ref{theta}) and some $r > 0$,
$\delta > 0$, and $\sigma > 0$ fixed, while the inequality
$\delta' \equiv \delta + \alpha \sigma \equiv 2 r - \delta - \sigma^2 > 0$ be satisfied.
Then, the optimal exercise times for the perpetual American standard and lookback
put and call options with the values in (\ref{VU5a1}) and (\ref{VU5c1}) and the payoff
functions $G_i(x, s)$ and $F_i(x, q)$ given after (\ref{VU5a1}) have the structure:
\begin{equation}
\label{tau*}
\tau^*_i = \inf \big\{t \ge 0 \; \big| \; X_t \le a^*_{i,\Xi^1_t}(S_t) \big\}
\quad \text{and} \quad
\zeta^*_i = \inf \big\{t \ge 0 \; \big| \; X_t \ge b^*_{i,\Xi^2_t}(Q_t) \big\}
\end{equation}
under $\alpha < 0$ and $\alpha > 0$, for $i = 1, 2, 3$, respectively,
where the optimal exercise boundaries $a^*_{i,j}(s)$ and $b^*_{i,j}(q)$
in (\ref{tau*}) represent some functions satisfying the inequalities
$a^*_{i,j}(s) < g^*_{i}(s) \wedge {\overline a}_{i,j}(s) \wedge s$,
for $s > {\underline s}_{i,j}$,
and $b^*_{i,j}(q) > q \vee {\underline b}_{i,j}(q) \vee h^*_{i}(q)$,
for $0 < q < {\overline q}_{i,j}$, while the lower and upper bounds for
$g^*_i(s) \wedge a^*_{i,j}(s)$ and $b^*_{i,j}(q) \vee h^*_{i}(q)$ as well as
the values ${\underline s}_{i,j}$ and ${\overline q}_{i,j}$ specified below
are determined as follows, for every $i = 1, 2, 3$ and $j = 0, 1$:

(i') we have the equalities $g^*_{1}$ given by (\ref{gh1}),
${\overline a}_{1,0} = r L_1/\delta'$, $\underline{s}_{1,0} = L_1$, $g^*_{2}(s) = {\underline \lambda} s$
with ${\underline \lambda}$ being the unique solution of (\ref{h71la}),
${\overline a}_{2,0}(s) = r s/(\delta' L_2)$, for $s > {\underline s}_{2,0} = 0$, and
$g^*_{3}(s)$ is the maximal solution of (\ref{h73}), for each $s > {\underline s}_{3,0} = L_3$,
under $\alpha < 0$;

(i'') the value ${\overline a}_{i,1}(s)$ is the maximal solution of the arithmetic
equation $H_{i,1,1}(a_{i,1}(s), s) = 0$, for each $s > {\underline s}_{i,1}$ with
some ${\underline s}_{i,1}$ implicitly specified below, with $H_{i,1,1}(x, s)$
given by (\ref{H5b1}), for all $0 < x \le s$ and every $i = 1, 2$, under $\alpha < 0$;

(ii') we have the equalities $h^*_1$ given by (\ref{gh1}), ${\underline b}_{1,0} = r K_1/\delta'$,
$\overline{q}_{1,0} = K_1$, $h^*_{2}(q) = {\overline \nu} q$ with ${\overline \nu}$ being the unique solution
of (\ref{h71la}), ${\underline b}_{2,0}(q) = r q/(\delta' K_2)$, for $q < {\overline q}_{2,0} \equiv \infty$,
and $h^*_{3}(q)$ is the minimal solution of (\ref{h74}), for each $q < {\overline q}_{3,0} = K_3$, under $\alpha > 0$;

(ii'') the value ${\underline b}_{i,1}(q)$ is the minimal solution
of the arithmetic equation $H_{i,2,1}(b_{i,1}(q), q) = 0$, for each
$q < {\overline q}_{i,1}$ with some ${\overline q}_{i,1}$ implicitly
specified below, with $H_{i,2,1}(x, q)$ given by (\ref{H5b2}),
for all $0 < q \le x$ and every $i = 1, 2$, under $\alpha > 0$.
\end{theorem}

\begin{proof}
        It follows from the general results of the optimal stopping theory for Markov processes
        (cf., e.g. \cite[Chapter~I, Subsection~2.2]{PSbook})) that the continuation regions related
        to the problems in (\ref{VU5c1}) above have the form:
        \begin{align}
        \label{C71}
        &C^*_{i,1,j} = \big\{ (x, s) \in E_{1} \; \big| \; V^*_{i,j}(x, s) > G_i(x, s) \big\}
        \end{align}
        and
        \begin{align}
        \label{C72}
        &C^*_{i,2,j} = \big\{ (x, q) \in E_{2} \; \big| \; U^*_{i,j}(x, q) > F_i(x, q) \big\}
        \end{align}
while the appropriate stopping regions are given by:
        \begin{align}
        \label{D71}
        &D^*_{i,1,j} = \big\{ (x, s) \in E_{1} \; \big| \; V^*_{i,j}(x, s) = G_i(x, s) \big\}
        \end{align}
        and
        \begin{align}
        \label{D72}
        &D^*_{i,2,j} = \big\{ (x, q) \in E_{2} \; \big| \; U^*_{i,j}(x, q) = F_i(x, q) \big\}
        \end{align}
        for every $i = 1, 2, 3$ and $j = 0, 1$.
        It is seen from the results of Theorem~\ref{thm41} proved below that the value functions
        $V^*_{i,j}(x, s)$ and $U^*_{i,j}(x, q)$ are continuous, so that the sets $C^*_{i,k,j}$,
        for $k = 1, 2$, in (\ref{C71})-(\ref{C72}) are open, while the sets $D^*_{i,k,j}$,
        for $k = 1, 2$, in (\ref{D71})-(\ref{D72}) are closed, for every $i = 1, 2, 3$ and $j = 0, 1$
        (see Figures~1-8 below for the computer drawings of the continuation and stopping regions
        $C^*_{i,k,j}$ and $D^*_{i,k,j}$, for $i = 1, 2, 3$, $k = 1, 2$ and $j = 0, 1$).
        We now describe the structure of the continuation and stopping regions from
        (\ref{C71})-(\ref{C72}) and (\ref{D71})-(\ref{D72}) in the following parts of the proof.

        \vspace{3pt}

        {\bf (i)} Let us first show that the stopping regions $D^*_{i,k,j}$,
        for $k = 1, 2$, in (\ref{D71})-(\ref{D72}) are nonempty, for every
        $i = 1, 2, 3$ and $j = 0, 1$, respectively.
        For this purpose, we consider the optimal stopping problems with the
        value functions in (\ref{VU7a1}) below.
        Since the suprema in the expressions of (\ref{VU7a1}) are taken over
        the stopping times with respect to the filtration $\FF$ which is
        obviously included in $\GG^k$, for $k = 1, 2$, we may easily conclude
        that the equalities ${\overline V}_i(x, s) \le V^*_{i,j}(x, s)$ and
        ${\overline U}_i(x, q) \le U^*_{i,j}(x, q)$ hold for the value functions in (\ref{VU5c1}),
        for all $0 < x \le s$ and $0 < q \le x$, and every $i = 1, 2, 3$ and $j = 0, 1$.
        Then, the standard comparison arguments applied for the value
        functions mentioned above imply that, for any point $(x, s) \in D^*_{i,1,j}$,
        we have ${\overline V}_i(x, s) \le V^*_{i,j}(x, s) \le G_i(x, s)$, so that
        $(x, s) \in D_{i,1}'$ holds, while, for any point $(x, q) \in D^*_{i,2,j}$,
        we have ${\overline U}_i(x, q) \le U^*_{i,j}(x, q) \le F_i(x, q)$, we have
        so that $(x, q) \in D_{i,2}'$ holds.
        These facts mean that all the points $(x, s)$ of the stopping region $D^*_{i,1,j}$
        in (\ref{D71}) surely belong to the stopping region $D_{i,1}'$ in (\ref{D75}), while
        all the points $(x, q)$ of the stopping region $D^*_{i,2,j}$ in (\ref{D72})
        surely belong to the stopping region $D_{i,2}'$ in (\ref{D75}),
        for every $i = 1, 2, 3$ and $j = 0, 1$.


        \vspace{3pt}

       {\bf (ii)} Observe that, by virtue of properties of the running maximum $S$ and minimum $Q$ from
        (\ref{SQ4}) of the generalised geometric Brownian motion $X$ from (\ref{X4a})-(\ref{X4b})
        (cf., e.g. \cite[Subsection~3.3]{DSS}, \cite[Subsection~3.3]{Pmax},
        or the arguments of our previous paper \cite[Theorem~2.1]{GL2}),
               it follows that the points $(x, s)$ of the diagonal
       $d_{1} = \{ (x, s) \in E_{1} \, | \, x = s \}$ belong
       to the continuation region $C^*_{i,1,0}$ in (\ref{C71}),
       while the points $(x, q)$ of the diagonal
       $d_{2} = \{ (x, q) \in E_{2} \, | \, x = q \}$ belong
       to the continuation region $C^*_{i,2,0}$ from (\ref{C72}),
       for every $i = 1, 2, 3$.




        Furthermore, it follows from the structure of the functions $H_{i,1,j}(x, s)$ and
        $H_{i,2,j}(x, q)$, for $j = 0, 1$, in (\ref{H5b1}) and (\ref{H5b2}) that it is not optimal to exercise
        the perpetual American standard or lookback put and call options (under progressively enlarged
        observations) when $H_{i,1,\Xi^1_t}(X_t, S_t) \ge 0$ or $H_{i,2,\Xi^2_t}(X_t, Q_t) \ge 0$ holds,
        for any $t \ge 0$, for every $i = 1, 2, 3$, respectively.
        In other words, these facts mean that all the points of the set $\{ (x, s) \in E_{1} \, | \,
        H_{i,1,j}(x, s) \ge 0 \}$ belong to the continuation region $C^*_{i,1,j}$ in (\ref{C71}),
        while all the points of the set $\{ (x, q) \in E_{2} \, | \, H_{i,2,j}(x, q) \ge 0 \}$
        belong to the continuation region $C^*_{i,2,j}$ in (\ref{C72}),
        for every $i = 1, 2, 3$ and $j = 0, 1$, respectively.


        In particular, the inequality $H_{1,1,0}(x, s) \equiv \delta' x - r L_1 \ge 0$
        is satisfied if and only if ${\overline a}_{1,0} \equiv r L_1/\delta' \le x \le s$ holds,
        the inequality $H_{2,1,0}(x, s) \equiv \delta' L_2 x - r s \ge 0$ is satisfied if and only if
        ${\overline a}_{2,0}(s) \equiv r s/(\delta' L_2) \le x \le s$ holds, while the inequality
        $H_{3,1,0}(x, s) \equiv r (L_3 - s) \ge 0$ is satisfied if and only if $0 < x \le s \le L_3$ holds.
        In other words, the points $(x, s)$ such that ${\overline a}_{i,0}(s) \le x \le s$ and
        $s > {\underline s}_{i,0}$, with $s_{1,0} = L_1$, $s_{2,0} = 0$ and $s_{3,0} = L_3$,
        belong to the region $C^*_{i,1,0}$ in (\ref{C71}), for every $i = 1, 2, 3$.
        Moreover, the inequality $H_{1,2,0}(x, q) \equiv r K_1 - \delta' x \ge 0$ is satisfied
        if and only if $q \le x \le {\underline b}_{1,0} \equiv r K_1/\delta'$ holds, the inequality
        $H_{2,2,0}(x, q) \equiv r q - \delta' K_2 x \ge 0$ is satisfied if and only if
        $q \le x \le {\underline b}_{2,0}(q) \equiv r q/(\delta' K_2)$ holds,
        while the inequality $H_{3,2,0}(x, q) \equiv r (q - K_3) \ge 0$ is satisfied if and only if
        $x \ge q \ge K_3$ holds.
        In other words, the points $(x, q)$ such that $q \le x \le {\underline b}_{i,0}(q)$ and
        $q < {\overline q}_{i,0}$, with $q_{1,0} = K_1$, $q_{2,0} = \infty$ and $q_{3,0} = K_3$,
        belong to the region $C^*_{i,2,0}$ in (\ref{C72}), for every $i = 1, 2, 3$
        (see Figures~1-4 below for the computer drawings of the boundary estimates ${\overline a}_{1,0}$
        and ${\underline b}_{1,0}$ as well as ${\overline a}_{2,0}(s)$ and ${\underline b}_{2,0}(q)$).

        Furthermore, the inequality $H_{i,1,1}(x, s) \ge 0$ is satisfied if and only if ${\overline a}_{i,1}(s) \le x \le s$ holds,
        so that the points $(x, s)$ such that ${\overline a}_{i,1}(s) \le x \le s$ and $s > {\underline s}_{i,1}$ belong to the
        region $C^*_{i,1,1}$ in (\ref{C71}), for every $i = 1, 2, 3$, where we set ${\overline a}_{3,1}(s) [= a^*_{3,1}(s)] = s$.
        Moreover, the inequality $H_{i,2,1}(x, q) \ge 0$ is satisfied if and only if $q \le x \le {\underline b}_{i,1}$ holds,
        so that the points $(x, q)$ such that $q \le x \le {\underline b}_{i,1}(q)$ and $q < {\overline q}_{i,1}$ belong to the
        region $C^*_{i,2,1}$ in (\ref{C72}), for every $i = 1, 2$, where we set ${\underline b}_{3,1}(q) [= b^*_{3,1}(q)] = q$.

        \vspace{3pt}

        {\bf (iii)} On the one hand, we observe that, if we take some $(x, s) \in D^*_{i,1,j}$ from (\ref{D71}) such that $x > g^*_{i}(s)$
        with $g^*_{i}(s)$ specified in Corollary~5.1 below and use the fact that the process $(X, S)$ started at some $(x', s)$ such that $g^*_{i}(s) \le x' < x$ passes through the point $(x, s)$ before hitting the diagonal $d_{1} = \{ (x, s) \in E_{1} \, | \, x = s \}$,
        then we see that the equality in (\ref{VU5c1}) implies that $V^*_{i,j}(x', s) \le V^*_{i,j}(x, s) = G_i(x, s)$
        holds, so that $(x', s) \in D^*_{i,1,j}$, for every $i = 1, 2, 3$ and $j = 0, 1$.
        Moreover, if we take some $(x, q) \in D^*_{i,2,j}$ from (\ref{D72})
        such that $x < h^*_{i}(q)$ with $h^*_{i}(q)$ specified in Proposition~5.1 above
        and use the fact that the process $(X, Q)$ started at some $(x', q)$ such that
        $h^*_{i}(q) \ge x' > x$ passes through the point $(x, q)$ before hitting the diagonal
        $d_{2} = \{ (x, q) \in E_{2} \, | \, x = q \}$, then we see that the equality in (\ref{VU5c1})
        implies that $U^*_{i,j}(x', q) \le U^*_{i,j}(x, q) = F_i(x, q)$ holds,
        so that $(x', q) \in D^*_{i,2,j}$, for every $i = 1, 2, 3$ and $j = 0, 1$.

        On the other hand, if take some $(x, s) \in C^*_{i,1,j}$ from (\ref{C71})
        such that $x < {\overline a}_{i,j}(s)$ and use the fact that the process $(X, S)$
        started at $(x, s)$ passes through some point $(x'', s)$ such that
        $x < x'' \le {\overline a}_{i,j}(s)$ before hitting the diagonal $d_{1}$, when $j = 0$, or coming
        towards the diagonal $d_{1}$, when $j = 1$, then we see that the equality in (\ref{VU5c1}) yields
        that $V^*_{i,j}(x'', s) \ge V^*_{i,j}(x, s) > G_i(x, s)$ holds, so that
        $(x'', s) \in C^*_{i,1,j}$, for every $i = 1, 2, 3$ and $j = 0, 1$.
        Moreover, if we take some $(x, q) \in C^*_{i,2,j}$ from (\ref{C72})
        and use the fact that the process $(X, Q, \Xi^2)$ started at $(x, q)$ passes through
        some point $(x'', q)$ such that ${\underline b}_{i,j}(q) \le x'' < x$ before hitting
        the diagonal $d_{2}$, when $j = 0$, or coming towards the diagonal $d_{2}$, when $j = 1$,
        then the equality in (\ref{VU5c1}) yields that $U^*_{i,j}(x'', q) \ge U^*_{i,j}(x, q) > F_i(x, q)$
        holds, so that $(x'', q) \in C^*_{i,2,j}$, for every $i = 1, 2, 3$ and $j = 0, 1$.

        Hence, we may conclude that there exist functions $a^*_{i,j}(s)$ and $b^*_{i,j}(q)$
        satisfying the inequalities $a^*_{i,j}(s) < g^*_i(s) \wedge {\overline a}_{i,j}(s) \wedge s$,
        for all $s > {\underline s}_{i,j}$, and $b^*_{i,j}(q) > q \vee {\underline b}_{i,j}(q) \vee h^*_{i}(q)$,
        for all $q < {\overline q}_{i,j}$, as well as the equalities
        $a^*_{1,j}(s) = s$, $a^*_{3,j}(s) = 0$, for all $s \le {\underline s}_{i,j}$, and
        $b^*_{1,j}(q) = q$, $b^*_{3,j}(q) = \infty$, for all $q \ge {\overline q}_{i,j}$,
        such that the continuation regions $C^*_{i,k,j}$, for $k = 1, 2$, in (\ref{C71})-(\ref{C72}) have the form:
        \begin{align}
        \label{C73}
        &C^*_{i,1,j} = \big\{ (x, s) \in E_{1} \; \big| \; a^*_{i,j}(s) < x \le s \big\}
        \;\; \text{and} \;\;
        C^*_{i,2,j} = \big\{ (x, q) \in E_{2} \; \big| \; q \le x < b^*_{i,j}(q) \big\}
        \end{align}
        while the stopping regions $D^*_{i,k,j}$, for $k = 1, 2$, in (\ref{D71})-(\ref{D72}) are given by:
        \begin{align}
        \label{D73}
        &D^*_{i,1,j} = \big\{ (x, s) \in E_{1} \; \big| \; x \le a^*_{i,j}(s) \big\}
        \;\; \text{and} \;\;
        D^*_{i,2,j} = \big\{ (x, q) \in E_{2} \; \big| \; x \ge b^*_{i,j}(q) \big\}
        \end{align}
        for every $i = 1, 2, 3$ and $j = 0, 1$, respectively
        (see Figures~1-8 below for the computer drawings of the optimal stopping
        boundaries $a^*_{i,j}(s)$ and $b^*_{i,j}(q)$, for $i = 1, 2, 3$ and $j = 0, 1$). $\square$
\end{proof}

\newpage

      \subsection{The free-boundary problems.}
      By means of standard arguments based on the application of It\^o's formula,
      it is shown that the infinitesimal operator $\LL$ of the process $(X, S, \Xi^1)$
      or $(X, Q, \Xi^2)$ from (\ref{X4a})-(\ref{X4b}) and (\ref{dX4a})-(\ref{dX4b})
      with (\ref{SQ4})-(\ref{theta}) has the form:
       \begin{align}
       \label{LF}
       &\LL = \big( r - \delta + \big[ \; \text{either} \;\; \Psi \big( (s/x)^{\alpha}, j \big)
       \;\; \text{or} \;\; \Phi \big( (q/x)^{\alpha}, j \big) \; \big] \big)
       \, x \, \partial_x + \frac{\sigma^2 x^2}{2} \, \partial_{xx} \\
       \notag
       &\phantom{\LL = \;\:} \text{in} \;\; 0 < x < s \;\; \text{or} \;\; 0 < q < x \\
       \label{LFd}
       &\partial_s = 0 \;\; \text{at} \;\; 0 < x = s \;\;
       \text{or} \;\; \partial_q = 0 \;\; \text{at} \;\; 0 < x = q \;\;
       \text{with} \;\; j = 0
       \end{align}
       where the functions $\Psi((s/x)^{\alpha}, j)$ and $\Phi((q/x)^{\alpha}, j)$, for $j = 0, 1$,
       are given by (\ref{psi4}) and (\ref{xi4}), respectively
      (cf., e.g. \cite[Subsection~3.1]{Pmax}).
      In order to find analytic expressions for the unknown
      value functions $V^*_{i,j}(x, s)$ and $U^*_{i,j}(x, q)$ from (\ref{VU5c1})
      and the unknown boundaries $a^*_{i,j}(s)$ and $b^*_{i,j}(q)$ from (\ref{C73}) and (\ref{D73}),
      for every $i = 1, 2, 3$ and $j = 0, 1$,
      we apply the results of general theory for solving optimal stopping problems
      for Markov processes presented in
      \cite[Chapter~IV, Section~8]{PSbook} among others (cf. also
      \cite[Chapter~V, Sections~15-20]{PSbook} for optimal stopping problems
      for maxima processes and other related references).

      More precisely, for the original
      optimal stopping problems in (\ref{VU5c1}), we formulate the associated
      free-boundary problems (cf., e.g.
      \cite[Chapter~IV, Section~8]{PSbook}) and then verify in Theorem~\ref{thm41}
      below that the appropriate candidate solutions of the latter problems
      coincide with the solutions of the original problems.
      In other words, we reduce the optimal stopping problems of (\ref{VU5c1})
      to the following equivalent free-boundary problems:
       \begin{align}
       \label{LV1}
       &({\LL} V_{i,j} - r V_{i,j})(x, s) = 0 \;\; \text{for} \;\;
       (x, s, j) \in C_{i,1,j} \setminus \{ (x, s) \in E_{1} \; | \; x = s \} \\
       \label{LV2}
       &({\LL} U_{i,j} - r U_{i,j})(x, q) = 0 \;\; \text{for} \;\;
       (x, q, j) \in C_{i,2,j} \setminus \{ (x, q) \in E_{2} \; | \; x = q \} \\
       \label{CF}
       &V_{i,j}(x, s) \big|_{x = a_{i,j}(s)+} = G_i(x, s) \big|_{x = a_{i,j}(s)},
       \;\;
       U_{i,j}(x, q) \big|_{x = b_{i,j}(q)-} = F_i(x, q) \big|_{x = b_{i,j}(q)} \\
       \label{SF}
       &\partial_x V_{i,j}(x, s) \big|_{x = a_{i,j}(s)+} =
       \partial_x G_i(x, s) \big|_{x = a_{i,j}(s)}, \;\;
       \partial_x U_{i,j}(x, q) \big|_{x = b_{i,j}(q)-} =
       \partial_x F_i(x, q) \big|_{x = b_{i,j}(q)} \\
       \label{NR1}
       &\big( \partial_s V_{i,0}(x, s) - \big( V_{i,1}(x, s)
       - V_{i,0}(x, s) \big) \, (\alpha/s) \big) \big|_{x = s-} = 0,
       \;\; V_{i,1}(x, s) \big|_{x = s-} = V_{i,0}(x, s) \big|_{x = s-} \\
       \label{NR2}
       &\big( \partial_s U_{i,0}(x, q) - \big( U_{i,1}(x, q)
       - U_{i,0}(x, q) \big) \, (\alpha/q) \big) \big|_{x = q+} = 0,
       \;\; U_{i,1}(x, q) \big|_{x = q+} = U_{i,0}(x, q) \big|_{x = q+} \!\!\! \\
       \label{VD}
       &V_{i,j}(x, s) = G_i(x, s) \;\; \text{for} \;\; (x, s) \in D_{i,1,j}, \;\;
       U_{i,j}(x, q) = F_i(x, q) \;\; \text{for} \;\; (x, q) \in D_{i,2,j} \\
       \label{VC}
       &V_{i,j}(x, s) > G_i(x, s) \;\; \text{for} \;\; (x, s) \in C_{i,1,j},
       \;\;
       U_{i,j}(x, q) > F_i(x, q) \;\; \text{for} \;\; (x, q) \in C_{i,2,j} \\
       \label{LVD1}
       &({\LL} V_{i,j} - r V_{i,j})(x, s) < 0 \;\; \text{for} \;\; (x, s) \in D_{i,1,j} \\
       \label{LVD2}
       &({\LL} U_{i,j} - r U_{i,j})(x, q) < 0 \;\; \text{for} \;\; (x, q) \in D_{i,2,j}
       \end{align}
      where the instantaneous-stopping and the smooth-fit as well as normal-reflection
      and normal-entrance conditions of (\ref{CF})-(\ref{NR2}) are satisfied,
      for all either $s > {\underline s}_{i,j}$ or $q < {\overline q}_{i,j}$,
      as well as every $i = 1, 2, 3$ and $j = 0, 1$, respectively.
      Here, the regions $C_{i,k,j}$ and $D_{i,k,j}$ are defined as $C^*_{i,k,j}$ and
      $D^*_{i,k,j}$, for $i = 1, 2, 3$, $k = 1, 2$ and $j = 0, 1$, in (\ref{C73})-(\ref{D73})
      with the unknown functions $a_{i,j}(s)$ and $b_{i,j}(q)$ instead of $a^*_{i,j}(s)$ and $b^*_{i,j}(q)$,
      while the functions $H_{i,1,j}(x, s)$ and $H_{i,2,j}(x, q)$, for every $i = 1, 2, 3$ and $j = 0, 1$,
      are defined in (\ref{H5b1}) and (\ref{H5b2}), respectively. Observe that the superharmonic characterisation of the
value function (cf., e.g. \cite[Chapter~IV, Section~9]{PSbook})
implies that $V^*_{i,j}(x, s)$ and $U^*_{i,j}(x, q)$ are the smallest functions satisfying
(\ref{LV1})-(\ref{CF}) and (\ref{VD})-(\ref{VC}) with the boundaries $a^*_{i,j}(s)$
and $b^*_{i,j}(q)$, for every $i = 1, 2, 3$ and $j = 0, 1$, respectively.
Note that the inequalities in (\ref{LVD1}) and (\ref{LVD2}) follow directly
from the arguments of Parts~(ii)-(iii) of Proof of Theorem~\ref{lem21} above.

\newpage

\begin{picture}(160,110)
\put(20,15){\begin{picture}(120,95) \label{fig1}
\put(0,0){\line(1,0){120}} \put(120,0){\line(0,1){95}}
\put(0,0){\line(0,1){95}}
\put(0,95){\line(1,0){120}}   
\put(10,10){\vector(1,0){100}}
\put(10,10){\vector(0,1){80}}    
\put(10,10){\line(1,1){75}}
\put(5,88){$s$} \put(108,5){$x$} \put(31,5){${a^*_{1,0}}$} \put(46,5){$g^*_1$} \put(55,5){${\bar a}_{1,0}$}
\put(79,75){$d_1 = \{ x = s \}$} \put(58,9){\line(0,1){2}} \put(48,9){\line(0,1){2}} \put(35,9){\line(0,1){2}}
\put(17,44){$D^*_{1,1,0}$} \put(37,67){$C^*_{1,1,0}$}
\put(58,58){\line(0,1){27}} \put(48,48){\line(0,1){37}} \put(35,35){\line(0,1){50}}
\put(0,0.1){\put(58,58){\line(0,1){27}} \put(48,48){\line(0,1){37}} \put(35,35){\line(0,1){50}}}
\put(0,-0.1){\put(58,58){\line(0,1){27}} \put(48,48){\line(0,1){37}} \put(35,35){\line(0,1){50}}}
\qbezier[48](58,10)(58,34)(58,58) \qbezier[38](48,10)(48,29)(48,48) \qbezier[25](35,10)(35,23)(35,35)
\end{picture}}
\put(20,10){\small{{\bf Figure~1.} {\rm A computer drawing of the optimal exercise boundary $a^*_{1,0}$.}}}
\end{picture}


\begin{picture}(160,110)
\put(20,15){\begin{picture}(120,95) \label{fig2}
\put(0,0){\line(1,0){120}} \put(120,0){\line(0,1){95}}
\put(0,0){\line(0,1){95}}
\put(0,95){\line(1,0){120}}   
\put(10,10){\vector(1,0){100}}
\put(10,10){\vector(0,1){80}}    
\put(10,10){\line(1,1){75}}
\put(60,10){\line(0,1){50}} \put(34,10){\line(0,1){24}} \put(44,10){\line(0,1){34}}
\put(0,0.1){\put(60,10){\line(0,1){50}} \put(34,10){\line(0,1){24}} \put(44,10){\line(0,1){34}}}
\put(0,-0.1){\put(60,10){\line(0,1){50}} \put(34,10){\line(0,1){24}} \put(44,10){\line(0,1){34}}}
\put(5,88){$q$} \put(108,5){$x$} \put(42,5){$h^*_1$}
\put(32,5){${\underline b}_{1,0}$} \put(60,9){\line(0,1){2}} \put(58,5){$b^*_{1,0}$}
\put(79,75){$d_2 = \{ x = q \}$} \put(44,9){\line(0,1){2}}
\put(34,9){\line(0,1){2}} \put(73,49){$C^*_{1,2,0}$} \put(48,22){$D^*_{1,2,0}$}
\end{picture}}
\put(20,10){\small{{\bf Figure~2.} {\rm A computer drawing of the
optimal exercise boundary $b^*_{1,0}$.}}}
\end{picture}

\newpage

\begin{picture}(160,110)
\put(20,15){\begin{picture}(120,95) \label{fig3}
\put(0,0){\line(1,0){120}} \put(120,0){\line(0,1){98}} \put(0,0){\line(0,1){98}}
\put(0,98){\line(1,0){120}}   
\put(10,10){\vector(1,0){100}}
\put(10,10){\vector(0,1){83}}    
\put(10,10){\line(1,1){80}}   
\put(10,10){\line(1,2){38}} \put(10,10){\line(1,4){19}} \put(10,10){\line(3,4){57}}
\put(0,0.1){\put(10,10){\line(1,2){38}} \put(10,10){\line(1,4){19}} \put(10,10){\line(3,4){57}}}
\put(0,-0.1){\put(10,10){\line(1,2){38}} \put(10,10){\line(1,4){19}} \put(10,10){\line(3,4){57}}}
\put(5,91){$s$} \put(108,5){$x$} \put(35,81){$g^*_2(s)$} \put(63,60){$d_{1} = \{ x = s \}$} \put(51,81){${\bar a}_{2,0}(s)$}
\put(15,81){$a^*_{2,0}(s)$} \put(28,70){$C^*_{2, 1, 0}$} \put(13,70){$D^*_{2, 1, 0}$}
\end{picture}}
\put(18,10){\small{{\bf Figure~3.} {\rm A computer drawing of the optimal exercise boundary $a^*_{2,0}(s)$.}}}
\end{picture}

\begin{picture}(160,110)
\put(20,15){\begin{picture}(120,95) \label{fig4}
\put(0,0){\line(1,0){120}} \put(120,0){\line(0,1){95}}
\put(0,0){\line(0,1){95}}
\put(0,95){\line(1,0){120}}   
\put(10,10){\vector(1,0){100}}
\put(10,10){\vector(0,1){80}}    
\put(10,10){\line(1,1){80}}   
\put(10,10){\line(2,1){77}} \put(10,10){\line(4,1){77}} \put(10,10){\line(4,3){77}}
\put(0,0.1){\put(10,10){\line(2,1){77}} \put(10,10){\line(4,1){77}} \put(10,10){\line(4,3){77}}}
\put(0,-0.1){\put(10,10){\line(2,1){77}} \put(10,10){\line(4,1){77}} \put(10,10){\line(4,3){77}}}
\put(5,88){$q$} \put(108,5){$x$} \put(75,30){${\underline b}_{2,0}(q)$} \put(75,68){$b^*_{2,0}(q)$}
\put(60,57){$C^*_{2,2,0}$} \put(60,43){$D^*_{2,2,0}$} \put(49,75){$d_{2} = \{ x = q \}$} \put(77,49){$h^*_2(q)$}
\end{picture}}
\put(18,10){\small{{\bf Figure~4.} {\rm A computer drawing of the optimal exercise boundary $b^*_{2,0}(q)$.}}}
\end{picture}

\newpage

\begin{picture}(160,110)
\put(20,15){\begin{picture}(120,95) \label{fig5}
\put(0,0){\line(1,0){120}} \put(120,0){\line(0,1){95}}
\put(0,0){\line(0,1){95}}
\put(0,95){\line(1,0){120}}   
\put(10,10){\vector(1,0){100}}
\put(10,10){\vector(0,1){80}}    
\put(10,10){\line(1,1){80}}
\qbezier(10,10)(40,45)(58,85) \qbezier(10,10)(30,45)(38,85)
\put(0,0.1){\qbezier(10,10)(40,45)(58,85) \qbezier(10,10)(30,45)(38,85)}
\put(0,-0.1){\qbezier(10,10)(40,45)(58,85) \qbezier(10,10)(30,45)(38,85)}
\put(24,79){$a^*_{2,1}(s)$} \put(43,79){${\bar a}_{2,1}(s)$} \put(63,60){$d_{1} = \{ x = s \}$}
\put(18,60){$D^*_{2,2,1}$} \put(34,60){$C^*_{2,2,1}$} \put(5,88){$s$} \put(108,5){$x$}
\end{picture}}
\put(19,10){\small{{\bf Figure 5.} {\rm A computer drawing of the
optimal exercise boundary $a^*_{2,1}(s)$.}}}
\end{picture}

\begin{picture}(160,110)
\put(20,15){\begin{picture}(120,95) \label{fig6}
\put(0,0){\line(1,0){120}} \put(120,0){\line(0,1){95}}
\put(0,0){\line(0,1){95}}
\put(0,95){\line(1,0){120}}   
\put(10,10){\vector(1,0){100}}
\put(10,10){\vector(0,1){80}}    
\put(10,10){\line(1,1){80}} \put(49,75){$d_{2} = \{ x = q \}$}
\qbezier(10,10)(45,40)(85,58)
\qbezier(10,10)(45,30)(85,38)
\put(0,0.1){\qbezier(10,10)(45,40)(85,58) \qbezier(10,10)(45,30)(85,38)}
\put(0,-0.1){\qbezier(10,10)(45,40)(85,58) \qbezier(10,10)(45,30)(85,38)}
\put(75,60){$b^*_{2,1}(q)$} \put(75,40){${\underline b}_{2,1}(q)$}
\put(54,36){$D^*_{2,2,1}$} \put(54,50){$C^*_{2,2,1}$} \put(5,88){$q$} \put(108,5){$x$}
\end{picture}}
\put(19,10){\small{{\bf Figure 6.} {\rm A computer drawing of the
optimal exercise boundary $b^*_{2,1}(q)$.}}}
\end{picture}

   \newpage

\begin{picture}(160,110)
\put(20,15){\begin{picture}(120,95) \label{fig7}
\put(0,0){\line(1,0){120}} \put(120,0){\line(0,1){95}}
\put(0,0){\line(0,1){95}}
\put(0,95){\line(1,0){120}}   
\put(10,10){\vector(1,0){100}}
\put(10,10){\vector(0,1){80}}    
\put(10,10){\line(1,1){80}} \put(63,60){$d_{1} = \{ x = s \}$}
\qbezier(10,25)(27,35)(43,85) \qbezier(10,25)(30,30)(65,85)
\put(0,0.1){\qbezier(10,25)(27,35)(43,85) \qbezier(10,25)(30,30)(65,85)}
\put(0,-0.1){\qbezier(10,25)(27,35)(43,85) \qbezier(10,25)(30,30)(65,85)}
\put(29,79){$a^*_{3,0}(s)$} \put(51,79){$g^*_{3}(s)$} \put(9,25){\line(1,0){2}} \put(4,23){$L_3$}
\put(18,60){$D^*_{3,1,0}$} \put(36,60){$C^*_{3,1,0}$} \put(5,88){$s$} \put(108,5){$x$}
\end{picture}}
\put(19,10){\small{{\bf Figure 7.} {\rm A computer drawing of the
optimal exercise boundary $a^*_{3,0}(s)$.}}}
\end{picture}

\begin{picture}(160,110)
\put(20,15){\begin{picture}(120,95) \label{fig8}
\put(0,0){\line(1,0){120}} \put(120,0){\line(0,1){95}}
\put(0,0){\line(0,1){95}}
\put(0,95){\line(1,0){120}}   
\put(10,10){\vector(1,0){100}}
\put(10,10){\vector(0,1){80}}    
\put(10,10){\line(1,1){80}} \put(49,75){$d_{2} = \{ x = q \}$}
\qbezier(25,10)(40,30)(85,68) \qbezier(25,10)(35,25)(85,43)
\put(0,0.1){\qbezier(25,10)(40,30)(85,68) \qbezier(25,10)(35,25)(85,43)}
\put(0,-0.1){\qbezier(25,10)(40,30)(85,68) \qbezier(25,10)(35,25)(85,43)}
\put(80,60){$b^*_{3,0}(q)$} \put(80,38){$h^*_{3}(q)$} \put(25,9){\line(0,1){2}} \put(23,5){$K_3$}
\put(53,36){$D^*_{3,2,0}$} \put(53,50){$C^*_{3,2,0}$} \put(5,88){$q$} \put(108,5){$x$}
\end{picture}}
\put(19,10){\small{{\bf Figure 8.} {\rm A computer drawing of the
optimal exercise boundary $b^*_{3,0}(q)$.}}}
\end{picture}

\newpage

    \section{\dot Solutions to the free-boundary problems}

       In this section, we obtain solutions to the free-boundary problems in (\ref{LV1})-(\ref{LVD2})
       in the form of analytic expressions for the candidate value functions through
       the candidate boundaries, by considering the cases $j = 0$ and $j = 1$, separately.
       In the first case, we derive first-order nonlinear ordinary differential equations
       for the candidate optimal stopping boundaries for the underlying risky asset price
       depending on the running values of the maximum and minimum processes, respectively.
       In the second case, we derive transcendental arithmetic equations for the candidate
       boundaries depending on the observed values of the global maximum and minimum of the
       underlying asset price, respectively.


       \subsection{The candidate value functions under $j = 0$, for $i = 1, 2, 3$.}
       In this case, it is shown that the second-order ordinary differential
       equations in (\ref{LF})+(\ref{LV1})-(\ref{LV2}) have the general solutions:
       \begin{align}
       \label{V312c}
       &V_{i,0}(x, s) = C_{i,1,0}(s) \, x^{\beta_1} + C_{i,2,0}(s) \, x^{\beta_2}
       \quad \text{and} \quad
       U_{i,0}(x, q) = D_{i,1,0}(q) \, x^{\beta_1} + D_{i,2,0}(q) \, x^{\beta_2}
       \end{align}
       when $\alpha < 0$, for $0 < x \le s$, and $\alpha > 0$, for $0 < q \le x$, respectively.
       Here, $C_{i,k,0}(s)$ and $D_{i,k,0}(q)$, for $i = 1, 2, 3$ and $k = 1, 2$,
       are some arbitrary (continuously differentiable)
       functions, and $\beta_l$, for $l = 1, 2$, are given by:
       \begin{equation}
       \label{beta12}
       \beta_l = \frac{1}{2}-\frac{r-\delta'}{\sigma^2} - (-1)^l
       \sqrt{\bigg( \frac{1}{2} - \frac{r-\delta'}{\sigma^2} \bigg)^2
       + \frac{2r}{\sigma^2}}
       \end{equation}
       with $\delta' \equiv \delta + \alpha \sigma \equiv 2 r - \delta - \sigma^2$,
       so that $\beta_2 < 0 < 1 < \beta_1$ holds.
       Then, by applying the conditions of (\ref{CF}) and the left-hand sides
       of (\ref{NR1}) and (\ref{NR2}) to the functions in (\ref{V312c}), we obtain the equalities:
       \begin{align}
       \label{C32a}
       &C_{i,1,0}(s) \, a^{\beta_1}_{i,0}(s) + C_{i,2,0}(s) \, a^{\beta_2}_{i,0}(s)
       = G_i(a_{i,0}(s), s) \\
       \label{C32b}
       &\beta_1 \, C_{i,1,0}(s) \, a^{\beta_1}_{i,0}(s) + \beta_2 \, C_{i,2,0}(s) \,
       a^{\beta_2}_{i,0}(s) = a_{i,0}(s) \, \partial_x G_i(a_{i,0}(s), s) \\
       \label{C32e}
       &C_{i,1,0}'(s) \, s^{\beta_1} + C_{i,2,0}'(s) \, s^{\beta_2} = 0
       \end{align}
       for all $s > {\underline s}_{i,0}$, and
       \begin{align}
       \label{C32c}
       &D_{i,1,0}(q) \, b^{\beta_1}_{i,0}(q) + D_{i,2,0}(q) \, b^{\beta_2}_{i,0}(q)
       = F_i(b_{i,0}(q), q) \\
       \label{C32d}
       &\beta_1 \, D_{i,1,0}(q) \, b^{\beta_1}_{i,0}(q) + \beta_2 \, D_{i,2,0}(q) \, b^{\beta_2}_{i,0}(q)
       = b_{i,0}(q) \, \partial_x F_i(b_{i,0}(q), q) \\
       \label{C32f}
       &D_{i,1,0}'(q) \, q^{\beta_1} + D_{i,2,0}'(q) \, q^{\beta_2} = 0
       \end{align}
       for all $q < {\overline q}_{i,0}$, respectively.
       Hence, by solving the systems of equations in (\ref{C32a})-(\ref{C32b})
       and (\ref{C32c})-(\ref{C32d}), we obtain that the candidate value functions
       admit the representations:
       \begin{equation}
       \label{Vg32a}
       V_{i,0}(x, s; a_{i,0}(s)) = C_{i,1,0}(s; a_{i,0}(s)) \, x^{\beta_1}
       + C_{i,2,0}(s; a_{i,0}(s)) \, x^{\beta_2}
       \end{equation}
       for $a_{i,0}(s) < x \le s$ and $s > {\underline s}_{i,0}$, with
       \begin{align}
       \label{CDg32a}
       &C_{i,l,0}(s; a_{i,0}(s)) =
       \frac{\beta_{3-l} G_i(a_{i,0}(s), s) - a_{i,0}(s) \partial_x G_i(a_{i,0}(s), s)}
       {(\beta_{3-l} - \beta_{l}) a^{\beta_l}_{i,0}(s)}
       \end{align}
       for $l = 1, 2$, and
       \begin{equation}
       \label{Vg32b}
       U_{i,0}(x, q; b_{i,0}(q)) = D_{i,1,0}(q; b_{i,0}(q)) \, x^{\beta_1}
       + D_{i,2,0}(q; b_{i,0}(q)) \, x^{\beta_2}
       \end{equation}
       for $q \le x < b_{i,0}(q)$ and $q < {\overline q}_{i,0}$, with
       \begin{align}
       \label{CDg32b}
       &D_{i,l,0}(q; b_{i,0}(q)) =
       \frac{\beta_{3-l} F_i(b_{i,0}(q), q) - b_{i,0}(q) \partial_x F_i(b_{i,0}(q), q)}
       {(\beta_{3-l} - \beta_l) b^{\beta_l}_{i,0}(q)}
       \end{align}
       for $i = 1, 2, 3$ and $l = 1, 2$, respectively.

Moreover, by means of straightforward computations, it can be deduced from the expressions
in (\ref{Vg32a}) and (\ref{Vg32b}) that the first-order and second-order partial derivatives
$\partial_x V_{i,0}(x, s; a_{i,0}(s))$ and $\partial_{xx} V_{i,0}(x, s; a_{i,0}(s))$ of
the function $V_{i,0}(x, s; a_{i,0}(s))$ take the form:
\begin{align}
\label{H'4a}
\partial_x V_{i,0}(x, s; a_{i,0}(s)) &=
C_{i,1,0}(s; a_{i,0}(s)) \, \beta_1 \, x^{\beta_1-1}
+ C_{i,2,0}(s; a_{i,0}(s)) \, \beta_2 \, x^{\beta_2-1}
\end{align}
and
\begin{align}
\label{H''4a}
\partial_{xx} V_{i,0}(x, s; a_{i,0}(s)) &=
C_{i,1,0}(s; a_{i,0}(s)) \, \beta_1 (\beta_1 - 1) \, x^{\beta_1-2}
+ C_{i,2,0}(s; a_{i,0}(s)) \, \beta_2 (\beta_2 - 1) \, x^{\beta_2-2} \!\!\!
\end{align}
on the interval $a_{i,0}(s) < x \le s$, for each $s > {\underline s}_{i,0}$
and every $i = 1, 2, 3$ fixed, while
the first-order and second-order partial derivatives
$\partial_x U_{i,0}(x, q; b_{i,0}(q))$ and $\partial_{xx} U_{i,0}(x, q; b_{i,0}(q))$
of the function $U_{i,0}(x, q; b_{i,0}(q))$ take the form:
\begin{align}
\label{H'4b}
\partial_x U_{i,0}(x, q; b_{i,0}(q)) &=
D_{i,1,0}(q; b_{i,0}(q)) \, \beta_1 \, x^{\beta_1-1}
+ D_{i,2,0}(q; b_{i,0}(q)) \, \beta_2 \, x^{\beta_2-1}
\end{align}
and
\begin{align}
\label{H''4b}
\partial_{xx} U_{i,0}(x, q; b_i(q)) &=
D_{i,1,0}(q; b_{i,0}(q)) \, \beta_1 (\beta_1 - 1) \, x^{\beta_1-2}
+ D_{i,2,0}(q; b_{i,0}(q)) \, \beta_2 (\beta_2 - 1) \, x^{\beta_2-2}
\end{align}
on the interval $q \le x < b_{i,0}(q)$, for each $q < {\overline q}_{i,0}$,
and every $i = 1, 2, 3$ fixed.



\subsection{The candidate optimal stopping boundaries under $j = 0$, for $i = 1, 2, 3$.}
       It follows that, for $i = 1$, we have
       \begin{equation}
       \label{ab01}
       a^*_{1,0} = \frac{\beta_2 L_1}{\beta_2 - 1}
       \quad \text{and} \quad
       b^*_{1,0} = \frac{\beta_1 K_1}{\beta_1 - 1}
       \end{equation}
       which satisfies $a^*_{1,0} < g^*_1 \wedge {\overline a}_{1,0}$ with ${\overline a}_{1,0} \equiv r L_1/\delta'$
       and $b^*_{1,0} > {\underline b}_{1,0} \vee h^*_1$ with ${\underline b}_{1,0} \equiv r K_1/\delta'$, where
       $g^*_1$ and $h^*_1$ are given in (\ref{gh1}) below.

       By applying the conditions of (\ref{C32e}) and (\ref{C32f}) to the functions in
       (\ref{CDg32a}) and (\ref{CDg32b}), we conclude that the candidate boundaries
       $a_{i,0}(s)$ and $b_{i,0}(q)$ satisfy the first-order nonlinear ordinary
       differential equations:
       \begin{align}
       \label{dg32a}
       &a_{i,0}'(s) =
       - \frac{\partial_s C_{i,1,0}(s; a_{i,0}(s)) s^{\beta_1}
       + \partial_s C_{i,2,0}(s; a_{i,0}(s)) s^{\beta_2}}
       {\partial_{a_{i,0}} C_{i,1,0}(s; a_{i,0}(s)) s^{\beta_1} +
       \partial_{a_{i,0}} C_{i,2,0}(s; a_{i,0}(s)) s^{\beta_2}}
       \end{align}
       for $s > {\underline s}_{i,0}$, and
       \begin{align}
       \label{dg32b}
       &b_{i,0}'(q) =
       - \frac{\partial_q D_{i,1,0}(q; b_{i,0}(q)) q^{\beta_1}
       + \partial_q D_{i,2,0}(q; b_{i,0}(q)) q^{\beta_2}}
       {\partial_{b_{i,0}} D_{i,1,0}(q; b_{i,0}(q)) q^{\beta_1} +
       \partial_{b_{i,0}} D_{i,2,0}(s; b_{i,0}(s)) q^{\beta_2}}
       \end{align}
       for $q < {\overline q}_{i,0}$, for $i = 2, 3$, respectively.

       In particular, for $i = 2$, it is seen from the expressions in (\ref{dg32a})
       and (\ref{dg32b}) that $a^*_{2,0}(s) \equiv \lambda_* s$,
       for $s > {\underline s}_{2} \equiv 0$, and
       $b^*_{2,0}(q) \equiv \nu_* q$,
       for $q < \infty \equiv {\overline q}_{2}$, where the numbers $0 < \lambda_* < 1$
       and $\nu_* > 1$ provide the unique roots of the power arithmetic equations:
       \begin{align}
       \label{g71la}
       &\lambda^{\beta_{1} - \beta_{2}} =
       \frac{(\beta_1 - 1) (\beta_2 (1 - L_2 \lambda) + L_2 \lambda)}
       {(\beta_2 - 1) (\beta_1 (1 - L_2 \lambda) + L_2 \lambda)}
       \quad \text{and} \quad
       \nu^{\beta_{1} - \beta_{2}} =
       \frac{(\beta_1 - 1) (\beta_2 (1 - K_2 \nu) + K_2 \nu)}
       {(\beta_2 - 1) (\beta_1 (1 - K_2 \nu) + K_2 \nu)}
       \end{align}
       on the intervals $(0, 1)$ and $(1, \infty)$, respectively.
       Furthermore, for $i = 3$, by means of straightforward computations,
       it follows that the first-order nonlinear ordinary differential equations
       in (\ref{dg32a}) and (\ref{dg32b}) take the form:
       \begin{align}
       \label{g73}
       a_{3,0}'(s) &= \frac{a_{3,0}(s)}{s - L_3} \,
       \frac{\beta_2 (s/a^{\beta_1}_{3,0}(s)) - \beta_1 (s/a^{\beta_2}_{3,0}(s))}
       {\beta_1 \beta_2 ((s/a^{\beta_1}_{3,0}(s)) - (s/a^{\beta_2}_{3,0}(s)))}
       \end{align}
       for all $s > {\underline s}_{3,0} \equiv L_3$, and
       \begin{align}
       \label{g74}
       b_{3,0}'(q) &= \frac{b_{3,0}(q)}{K_3 - q} \,
       \frac{\beta_2 (q/b^{\beta_1}_{3,0}(q)) - \beta_1 (q/b^{\beta_2}_{3,0}(q))}
       {\beta_1 \beta_2 ((q/b^{\beta_1}_{3,0}(q)) -(q/b^{\beta_2}_{3,0}(q)))}
       \end{align}
       for all $q < {\overline q}_{3,0} \equiv K_3$, respectively.

  \subsection{The maximal and minimal admissible solutions $a^*_{3,0}(s)$ and $b^*_{3,0}(q)$.}
        We further consider the {\it maximal and minimal admissible} solutions of first-order
        nonlinear ordinary differential equations as the largest and smallest possible solutions
        $a^*_{3,0}(s)$ and $b^*_{3,0}(q)$ of the equations in (\ref{g73}) and (\ref{g74})
        which satisfy the inequalities $a^*_{3,0}(s) < s$ and $b^*_{3,0}(q) > q$,
        for all $s > {\underline s}_{3,0} \equiv L_3$ and $q < {\overline q}_{3,0} \equiv K_3$.
        By virtue of the classical results on the existence and uniqueness of solutions
        for first-order nonlinear ordinary differential equations, we may conclude that
        these equations admit (locally) unique solutions, in view of the facts that
        the right-hand sides in (\ref{g73}) and (\ref{g74}) are (locally) continuous in
        $(s, a_{3,0}(s))$ and $(q, b_{3,0}(q))$
        and (locally) Lipschitz in $a_{3,0}(s)$ and $b_{3,0}(q)$, for each
        $s > {\underline s}_{3,0}$ and $q < {\overline q}_{3,0}$ fixed
        (cf. also \cite[Subsection~3.9]{Pmax} for similar arguments based
        on the analysis of other first-order nonlinear ordinary differential equations).
        Then, it is shown by means of technical arguments based on Picard's method of
        successive approximations that there exist unique solutions $a_{3,0}(s)$ and $b_{3,0}(q)$
        to the equations in (\ref{g73}) and (\ref{g74}), for $s > {\underline s}_{3,0}$
        and $q < {\overline q}_{3,0}$, started at some points $(a_{3,0}(s_{3,0}'), s_{3,0}')$ and
        $(b_{3,0}(q_{3,0}'), q_{3,0}')$ such that $s_{3,0}' > {\underline s}_{3,0}$ and
        $q_{3,0}' < {\overline q}_{3,0}$
        (cf. also \cite[Subsection~3.2]{GP1} and \cite[Example~4.4]{Pmax} for similar arguments
        based on the analysis of other first-order nonlinear ordinary differential equations).

Hence, in order to construct the appropriate functions $a^*_{3,0}(s)$ and $b^*_{3,0}(q)$
which satisfy the equations in (\ref{g73}) and (\ref{g74}) and stays strictly above
or below the appropriate diagonal, for $s > {\underline s}_{3,0}$ and
$q < {\overline q}_{3,0}$, and every $i = 2, 3$, respectively,
we can follow the arguments from \cite[Subsection~3.5]{Pe5b}
(among others) which are based on the construction of sequences
of the so-called bad-good solutions which intersect the upper
or lower bounds or diagonals.
For this purpose, for any sequences $(s_{3,0,m})_{m \in \NN}$ and $(q_{3,0,m})_{m \in \NN}$
such that $s_{3,0,m} > {\underline s}_{3,0}$ and $q_{3,0,m} < {\overline q}_{3,0}$ as well as
$s_{3,0,m} \uparrow \infty$ and $q_{3,0,m} \downarrow 0$ as $m \to \infty$,
we can construct the sequence of solutions $a_{3,0,m}(s)$ and $b_{3,0,m}(q)$,
for $m \in \NN$, to the equations (\ref{g73}) and (\ref{g74}),
for all $s > {\underline s}_{3,0}$ and $q < {\overline q}_{3,0}$ such that
$a_{3,0,m}(s_{3,0,m}) < s_{3,0,m}$ and $b_{i,0,m}(q_{i,0,m}) > q_{i,0,m}$ holds,
for each $m \in \NN$.
It follows from the structure of the equations in (\ref{g73}) and (\ref{g74})
that the inequalities $a_{3,0,l}'(s_{3,0,m}) < 1$ and $b_{3,0,l}'(q_{3,0,m}) < 1$
should hold for the derivatives of the appropriate functions, for each $m \in \NN$
(cf. also \cite[pages~979-982]{Jesper} for the analysis of solutions of another first-order nonlinear differential equation).
Observe that, by virtue of the uniqueness of solutions mentioned above, we know that each two curves
$s \mapsto a_{3,0,m}(s)$ and $s \mapsto a_{3,0,n}(s)$ as well as $q \mapsto b_{3,0,m}(q)$ and
$q \mapsto b_{3,0,n}(q)$ cannot intersect, for $m, n \in \NN$ and $m \neq n$,
and thus, we see that the sequence $(a_{3,0,m}(s))_{m \in \NN}$ is increasing
and the sequence $(b_{3,0,m}(q))_{m \in \NN}$ is decreasing,
so that the limits $a^*_{3,0}(s) = \lim_{m \to \infty} a_{3,0,m}(s)$ and
$b^*_{3,0}(q) = \lim_{m \to \infty} b_{3,0,m}(q)$ exist,
for each $s > {\underline s}_{3,0}$ and $q < {\overline q}_{3,0}$, respectively.
We may therefore conclude that $a^*_{3,0}(s)$ and $b^*_{3,0}(q)$ provides the maximal
and minimal solutions to the equations in (\ref{g73}) and (\ref{g74}) such that
$a^*_{3,0}(s) < s$ and $b^*_{3,0}(q) > q$ holds,
for all $s > {\underline s}_{3,0}$ and $q < {\overline q}_{3,0}$.

Moreover, since the right-hand sides of the first-order nonlinear ordinary differential equations
in (\ref{g73}) and (\ref{g74}) are (locally) Lipschitz in $s$ and $q$, respectively, one can deduce
by means of Gronwall's inequality that the functions $a_{3,0,m}(s)$ and $b_{3,0,m}(q)$, for $m \in \NN$,
are continuous, so that the functions $a^*_{3,0}(s)$ and $b^*_{3,0}(q)$ are continuous too.
        The appropriate {\it maximal admissible} solutions of first-order nonlinear ordinary differential
        equations and the associated maximality principle for solutions of optimal stopping problems which
        is equivalent to the superharmonic characterisation of the payoff functions were established
        in \cite{Pmax} and further developed in \cite{GP1}, \cite{Jesper}, \cite{GuoShepp}, \cite{Gap},
        \cite{BK3}, \cite{GuoZer}, \cite{Pe5a}-\cite{Pe5b}, \cite{GHP}, \cite{Ott}, \cite{KyprOtt},
        \cite{GR1}-\cite{GR3}, \cite{RodZer}, and \cite{GKL} among other subsequent papers
        (cf. also \cite[Chapter~I; Chapter~V, Section~17]{PSbook} for other references).


\subsection{The candidate value functions under $j = 1$, for $i = 1, 2, 3$.}
In this case, the ordinary differential equations (with parameters)
in (\ref{LF})+(\ref{LV1})-(\ref{LV2}) have the general solution to
that equation takes the form:
\begin{equation}
\label{hatU1}
V_{i,1}(x, s) = C_{i,1,1}(s) \, W_{1,1}(x, s) + C_{i,1,2}(s) \, W_{1,2}(x, s)
\end{equation}
and
\begin{equation}
\label{hatU2}
U_{i,1}(x, q) = D_{i,1,1}(q) \, W_{2,1}(x, q) + D_{i,1,2}(s) \, W_{2,2}(x, q)
\end{equation}
where $C_{i,1,k}(s)$ and $D_{i,1,k}(q)$, for $i = 1, 2, 3$ and $k = 1, 2$,
are some arbitrary continuous functions.
Here, we assume that the functions $W_{1,k}(x, s)$ and $W_{2,k}(x, q)$, for $k = 1, 2$,
represent fundamental solutions to the homogeneous second-order ordinary differential
equations related to the one of (\ref{LF})+(\ref{LV1})-(\ref{LV2}) given by:
\begin{equation}
\label{hatH1}
W_{1,k}(x, s) = x^{\gamma_{3-k}} \, \big( 1 - (s/x)^{\alpha} \big)^{1-2/\sigma}
F \big( \chi_{k,2}, \chi_{k,1}; \varkappa_{k}; (s/x)^{\alpha} \big)
\end{equation}
for each $0 < x \le s$, and
\begin{equation}
\label{hatH2}
W_{2,k}(x, q) = x^{\gamma_{3-k}} \, \big( 1 - (q/x)^{\alpha} \big)^{1-2/\sigma} F \big( \chi_{k,2}, \chi_{k,1}; \varkappa_{k};
(q/x)^{\alpha} \big)
\end{equation}
for each $0 < q \le x$, and every $k = 1, 2$.
Here, we denote by $F(\alpha, \beta; \gamma; x)$ the Gauss'
hypergeometric function which is defined by means of the expansion:
\begin{equation}
\label{F41b}
F(\alpha, \beta; \gamma; x) = 1 + \sum_{m = 1}^{\infty}
\frac{(\alpha)_m (\beta)_m}{(\gamma)_m} \, \frac{x^m}{m!}
\end{equation}
for all $\alpha, \beta, \gamma \in \RR$ such that $\gamma \neq 0, -1, -2, \ldots$,
and $(\gamma)_m = \gamma (\gamma+1) \cdots (\gamma+m-1)$, for $m \in \NN$
(cf., e.g. \cite[Chapter~XV]{AS} and \cite[Chapter~II]{BE}), and additionally set
\begin{equation}
\label{chi}
\chi_{k, l} = 1 + \frac{\gamma_{k} - \beta_l}{\alpha} \quad \text{and} \quad \varkappa_k = 1 + \frac{2}{\alpha} \,
\bigg( \gamma_{k} - \frac{1}{2} + \frac{r - \delta}{\sigma^2} \bigg)
\end{equation}
for every $k, l = 1, 2$, where $\beta_k$ and $\gamma_k$, for $k = 1, 2$, are given by (\ref{beta12}) above and (\ref{gamma12}) below.
Thus, $W_{1,k}(x, s)$ and $W_{2,k}(x, q)$, for $k = 1, 2$, in the expressions of
(\ref{hatH1}) and (\ref{hatH2}), under $j = 1$,
are (strictly) increasing and decreasing (convex) functions satisfying the properties
$W_{1,1}(0+, s) = W_{2,1}(0+, q) = + 0$, $W_{1,1}(\infty, s) = W_{2,1}(\infty, q) = \infty$
and $W_{1,2}(0+, s) = W_{2,2}(0+, q) = \infty$, $W_{1,2}(\infty, s) = W_{2,2}(\infty, q) = + 0$,
for each $s > 0$ and $q > 0$ fixed, respectively
(cf., e.g. \cite[Chapter~V, Section~50]{RW} for further details).

Finally, by applying the conditions of (\ref{CF}) and the right-hand side of (\ref{NR1})-(\ref{NR2})
to the functions in (\ref{hatU1}) and (\ref{hatU2}), we obtain that the equalities:
\begin{align}
\label{UCont1}
&C_{i,1,1}(s) \, W_{1,1}(a_{i,1}(s), s) + C_{i,1,2}(s) \, W_{1,2}(a_{i,1}(s), s) =
G_i(a_{i,1}(s), s) \\
\label{USmft1}
&C_{i,1,1}(s) \, \partial_x W_{1,1}(a_{i,1}(s), s) +
C_{i,1,2}(s) \, \partial_x W_{1,2}(a_{i,1}(s), s) = \partial_x G_i(a_{i,1}(s), s) \\
\label{NCont1}
&C_{i,1,1}(s) \, W_{1,1}(s, s) + C_{i,1,2}(s) \, W_{1,2}(s, s) = V_{i,0}(s, s)
\end{align}
should hold, for each $s > {\underline s}_{i,1}$ fixed, with ${\underline s}_{i,1}$
specified below, and
\begin{align}
\label{UCont2}
&D_{i,1,1}(q) \, W_{2,1}(b_{i,1}(q), q) + D_{i,1,2}(q) \, W_{2,2}(b_{i,1}(q), q)
= F_i(b_{i,1}(q), q) \\
\label{USmft2}
&D_{i,1,1}(q) \, \partial_x W_{2,1}(b_{i,1}(q), q) +
D_{i,1,2}(q) \, \partial_x W_{2,2}(b_{i,1}(q), q) = \partial_x F_i(b_{i,1}(q), q) \\
\label{NCont2}
&D_{i,1,1}(q) \, W_{2,1}(q, q) + D_{i,1,2}(q) \, W_{2,2}(q, q) = U_{i,0}(q, q)
\end{align}
should hold, for each $q < {\overline q}_{i,1}$ fixed, with ${\overline q}_{i,1}$
specified below, and every $i = 1, 2, 3$.
Then, solving the system of equations in (\ref{UCont1})-(\ref{USmft1})
and (\ref{UCont2})-(\ref{USmft2}), we obtain that the functions
\begin{equation}
\label{hatUpart1}
V_{i,1}(x, s; a_{i,1}(s)) = C_{i,1,1}(s; a_{i,1}(s)) \, W_{1,1}(x, s)
+ C_{i,1,2}(s; a_{i,1}(s)) \, W_{1,2}(x, s)
\end{equation}
for $a_{i,1}(s) < x < s$, with
\begin{align}
\label{hatC51}
&C_{i,1,l}(s; a_{i,1}(s)) \\
\notag
&= \frac{G_i(a_{i,1}(s), s) \partial_x W_{1,3-l}(a_{i,1}(s), s)
- \partial_x G_i(a_{i,1}(s), s) W_{1,3-l}(a_{i,1}(s), s)}
{W_{1,l}(a_{i,1}(s), s) \partial_x W_{1,3-l}(a_{i,1}(s), s)
- W_{1,3-l}(a_{i,1}(s), s) \partial_x W_{1,l}(a_{i,1}(s), s)}
\end{align}
for each $s > {\underline s}_{i,1}$ fixed, and every $l = 1, 2$, and
\begin{equation}
\label{hatUpart2}
U_{i,1}(x, q; b_{i,1}(q))
= D_{i,1,1}(q; b_{i,1}(q)) \, W_{2,1}(x, q) + D_{i,1,2}(q; b_{i,1}(q)) \, W_{2,2}(x, q)
\end{equation}
for $q < x < b_{i,1}(q)$, with
\begin{align}
\label{hatC52}
&D_{i,1,k}(q; b_{i,1}(q)) \\
\notag
&= \frac{F_i(b_{i,1}(q), q) \partial_x W_{2,3-l}(b_{i,1}(q), q)
- \partial_x F_i(b_{i,1}(q), q) W_{2,3-l}(b_{i,1}(q), q)}
{W_{2,l}(b_{i,1}(q), q) \partial_x W_{2,3-l}(b_{i,1}(q), q)
- W_{2,3-l}(b_{i,1}(q), q) \partial_x W_{2,l}(b_{i,1}(q), q)}
\end{align}
for each $q < {\overline q}_{i,1}$ fixed, and every $l = 1, 2$, for every $i = 1, 2, 3$.
Hence, putting the expressions from (\ref{hatC51})
into the system of equations in (\ref{NCont1}) and (\ref{NCont2}), we obtain that
the functions $a_{i,1}(s)$ and $b_{i,1}(q)$ in (\ref{hatUpart1}) and (\ref{hatUpart2})
satisfy the conditions of (\ref{NCont1}) and (\ref{NCont2}) when the equality:
\begin{align}
\label{hatUeq1}
&C_{i,1,1}(s; a_{i,1}(s)) \, W_{1,1}(s, s) + C_{i,1,2}(s; a_{i,1}(s)) \, W_{1,2}(s, s) = V_{i,0}(s, s)
\end{align}
holds with $V_{i,0}(s, s)$ given by (\ref{Vg32a}) above, for each $s > {\underline s}_{i,1}$ fixed,
and the equality:
\begin{align}
\label{hatUeq2}
&D_{i,1,1}(q; b_{i,1}(q)) \, W_{2,1}(q, q) + D_{i,2,1}(q; b_{i,1}(q)) \, W_{2,2}(q, q) = U_{i,0}(q, q)
\end{align}
holds with $U_{i,0}(q, q)$ given by (\ref{Vg32b}) above, for each $q < {\overline q}_{i,1}$ fixed, for every $i = 1, 2$, respectively.
We further consider the appropriate maximal and minimal solutions $a^*_{i,1}(s)$ and $b^*_{i,1}(q)$ of the
transcendental arithmetic equations from (\ref{hatUeq1}) and (\ref{hatUeq2}) such that the inequalities
in (\ref{VC})-(\ref{LVD2}) hold for the candidate value functions $V_{i,1}(x, s; a_{i,1}(s))$ and
$U_{i,1}(x, q; b_{i,1}(q))$ from (\ref{hatUpart1}) and (\ref{hatUpart2}), for every $i = 1, 2$, respectively.
It follows from the arguments of Proof of Theorem~\ref{lem21} that the latter inequalities take
the form $a^*_{i,1}(s) < g^*_i(s) \wedge {\overline a}_{i,1}(s) \wedge s$, for $s > {\underline s}_{i,1}$,
and $b^*_{i,1}(q) > h^*_i(q) \vee {\underline b}_{i,1}(q) \vee q$, for $q < {\overline q}_{i,1}$, with
${\underline s}_{i,1}$ and ${\overline q}_{i,1}$ such that the appropriate solutions of (\ref{hatUeq1})
and (\ref{hatUeq2}) exist, for every $i = 1, 2$, respectively.


    \section{\dot Main results and proofs}

     In this section, based on the expressions computed above, we formulate and prove the main results of the paper.
     Note that the upper and lower bounds ${\overline a}_{i,j}(s)$ and ${\underline b}_{i,j}(q)$
     for the optimal exercise boundaries $a^*_{i,j}(s)$ and $b^*_{i,j}(q)$, for every $i = 1, 2$
     and $j = 0, 1$, respectively, are specified in Theorem~\ref{lem21} above.

\begin{theorem} \label{thm41}
Let the processes $(X, S, \Xi^1 \equiv I(\theta \le \cdot))$ and $(X, Q, \Xi^2 \equiv I(\eta \le \cdot))$
be given by (\ref{X4a})-(\ref{X4b}) and (\ref{dX4a})-(\ref{dX4b}) with (\ref{SQ4})-(\ref{theta})
and some $r > 0$, $\delta > 0$, and $\sigma > 0$ fixed, while the inequality
$\delta' \equiv \delta + \alpha \sigma \equiv 2 r - \delta - \sigma^2 > 0$ be satisfied.
Then, the value functions of the perpetual American standard and lookback
put and call options from (\ref{VU5a1}) and (\ref{VU5c1}) and the payoff
functions $G_i(x, s)$ and $F_i(x, q)$ given after (\ref{VU5a1}) above
admit the representations:
\begin{equation}
\label{4V*4'}
      V^*_{i,j}(x, s) =
      \begin{cases}
      V_{i,0}(x, s; a^*_{i,0}(s)), & \text{if} \;\; a^*_{i,0}(s) < x \le s
      \;\; \text{and} \;\; s > {\underline s}_{i,j}, \;\; j = 0 \\
      V_{i,1}(x, s; a^*_{i,1}(s)), & \text{if} \;\; a^*_{i,1}(s) < x \le s
      \;\; \text{and} \;\; s > {\underline s}_{i,j}, \;\; j = 1 \\
      G_i(x, s), & \text{if either} \;\; 0 < x \le a^*_{i,j}(s) \;\; \text{and}
      \;\; s > {\underline s}_{i,j} \;\; \text{or} \;\; x \le s \le {\underline s}_{i,j} \!\!\!\!\!
      \end{cases}
\end{equation}
whenever $\alpha \equiv 2 (r - \delta)/{\sigma^2} - 1 < 0$, and
\begin{equation}
\label{rho3c}
      U^*_{i,j}(x, q)=
      \begin{cases}
      U_{i,0}(x, q; b^*_{i,0}(q)), & \text{if} \;\; q \le x < b^*_{i,0}(q)
      \;\; \text{and} \;\; q < {\overline q}_{i,j}, \;\; j = 0 \\
      U_{i,1}(x, q; b^*_{i,1}(q)), & \text{if} \;\; q \le x < b^*_{i,1}(q)
      \;\; \text{and} \;\; q < {\overline q}_{i,j}, \;\; j = 1 \\
      F_i(x, q), & \text{if either} \;\; x \ge b^*_{i,j}(q) \;\; \text{and} \;\;
      q < {\overline q}_{i,j} \;\; \text{and} \;\; x \ge q \ge {\overline q}_{i,j} \!\!\!\!\!
      \end{cases}
\end{equation}
whenever $\alpha > 0$, while the optimal exercise times
have the form of (\ref{tau*}) above, where the candidate value functions and exercise boundaries
are specified as follows:


(i) the functions $V_{i,0}(x, s; a_{i,0}(s))$ and $V_{i,1}(x, s; a_{i,1}(s))$
are given by (\ref{Vg32a}) and (\ref{hatUpart1}) with (\ref{CDg32a}), whenever $\alpha < 0$,
while the optimal exercise boundaries $a^*_{i,0}(s)$ ($a^*_{1,0}$ is given by (\ref{ab01}))
and $a^*_{i,1}(s)$ provide the maximal solution to either the first-order nonlinear ordinary
differential equation in (\ref{dg32a}) or the transcendental arithmetic equation in (\ref{hatUeq1})
satisfying the inequalities $a^*_{i,j}(s) < g^*_i(s) \wedge {\overline a}_{i,j}(s) \wedge s$,
for all $s > {\underline s}_{i,j}$ as well as every $i = 2, 3$ and $j = 0, 1$, respectively;


\newpage
(ii) the functions $U_{i,0}(x, q; b_{i,0}(q))$ and $U_{i,1}(x, q; b_{i,1}(q))$ are given
by (\ref{Vg32b}) and (\ref{hatUpart2}) with (\ref{CDg32b}), whenever $\alpha > 0$,
while the optimal exercise boundaries $b^*_{i,0}(q)$ ($b^*_{1,0}$ is given by (\ref{ab01}))
and $b^*_{i,1}(q)$ provide the minimal solution to either the first-order nonlinear
ordinary differential equation in (\ref{dg32b}) or the transcendental arithmetic equation
in (\ref{hatUeq2}) satisfying the inequalities
$b^*_{i,j}(q) > q \vee {\underline b}_{i,j}(q) \vee h^*_i(q)$, for all
$q < {\overline q}_{i,j}$ as well as every $i = 2, 3$ and $j = 0, 1$, respectively.
\end{theorem}

Since both parts of the assertion stated above are proved using similar arguments, we only
give a proof for the case of the left-hand three-dimensional optimal stopping problem of
(\ref{VU5c1}) related to the perpetual American standard and lookback put options.
Observe that we can put $s = x$ and $q = x$ as well as $j = 0$ to obtain the values
of the original perpetual American standard and lookback put and call option pricing
problems of (\ref{VU5a1}) from the values of the optimal stopping problems of (\ref{VU5c1}).

\begin{proof}
        In order to verify the assertion stated
        above, it remains for us to show that the function defined in (\ref{4V*4'}) coincides with
        the value function in (\ref{VU5c1}) and that the stopping time $\tau^*_i$ in (\ref{tau*})
        is optimal with the boundaries $(a^*_{i,0}(s), a^*_{i,1}(s))$, for every $i = 1, 2, 3$, specified above.
        For this purpose, let $a_{i,0}(s)$ be any solution to the ordinary differential equation
        in (\ref{dg32a}) and $a_{i,1}(s)$ be any solution to the arithmetic transcendental equation in
        (\ref{hatUeq1}) satisfying the inequalities $a_{i,j}(s) < g^*_i \wedge {\overline a}_{i,j}(s) \wedge s$,
        where ${\overline a}_{i,j}(s)$ and $g^*_i(s)$ are specified in Theorem~\ref{lem21} above and
        Proposition~\ref{CorA1} below, respectively, for $s > {\underline s}_{i,j}$ and every
        $i = 1, 2, 3$ and $j = 0, 1$.
        Let us also denote by $V^{a_i}_{i,j}(x, s)$  the right-hand side of
        the expression in (\ref{4V*4'}) associated with $a_{i,j}(s)$, for every $i = 1, 2, 3$ and $j = 0, 1$.
  Then, it is shown by means of straightforward calculations from the previous section that the function
  $V^{a_i}_{i,j}(x, s)$ solves the system of (\ref{LV1}) with (\ref{VD})-(\ref{LVD1})
  and satisfies the conditions of (\ref{CF})-(\ref{NR1}). Recall that the function
  $V^{a_i}_{i,j}(x, s)$ is of the class $C^{2,1}$ on the closure ${\bar C}_{i,1,j}$ of $C_{i,1,j}$ and is equal
  to $G_i(x, s)$ on $D_{i,1,j}$, which are defined as ${\bar C}^*_{i,1,j}$, $C^*_{i,1,j}$ and
  $D^*_{i,1,j}$ in (\ref{C73}) and (\ref{D73}) with $a_{i,j}(s)$ instead of
  $a^*_{i,j}(s)$, for every $i = 1, 2, 3$ and $j = 0, 1$, respectively.

        Hence, taking into account the assumption
        that the boundary $a_{i,j}(s)$ is continuously differentiable,
        for all $s > {\underline s}_{i,j}$, by applying the change-of-variable
        formula from \cite[Theorem~3.1]{Pe1a} to the process
        $(e^{- r t} V^{a_i}_{i,\Xi^1_t}(X_t, S_t))_{t \ge 0}$ (cf. also \cite[Chapter~II, Section~3.5]{PSbook}
        for a summary of the related results and further references), we obtain the expression:
        \begin{align}
        \label{rho4c}
        &e^{- r t} \, V^{a_i}_{i,\Xi^1_t}(X_t, S_t) = V^{a_i}_{i,j}(x, s) + M^{i,1}_t \\
        \notag
        &+ \int_0^t e^{- r u} \, \big( \LL V^{a_i}_{i,\Xi^1_u} - r V^{a_i}_{i,\Xi^1_u} \big)
        (X_{u}, S_{u}) \, I(X_{u} \neq a_{i,\Xi^1_u}(S_u), X_u \neq S_u) \, du \\
        \notag
        &+ \int_0^t e^{- r u} \, \bigg( \partial_s V^{a_i}_{i,0}(X_{u}, S_{u})
        - \big( V^{a_i}_{i,1}(X_{u}, S_{u}) - V^{a_i}_{i,0}(X_{u}, S_{u}) \big) \,
        \frac{\alpha}{S_u} \bigg) \, I \big( X_u = S_u, \Xi^1_u = 0 \big) \, dS_u
\end{align}
for all $t \ge 0$, for every $i = 1, 2, 3$.
Here, the process $M^{i,1} = (M^{i,1}_t)_{t \ge 0}$ defined by:
\begin{align}
\label{N5}
M^{i,1}_t &= \int_0^t e^{- r u} \, \partial_x V^{a_i}_{i,\Xi^1_u}(X_u, S_u) \,
I(X_{u} \neq S_{u}) \, \sigma \, X_u \, d{\tilde B}^1_u \\
\notag
&\phantom{=\;\:}+ \int_0^t e^{- r u} \,
\big( V^{a_i}_{i,1}(X_{u}, S_{u}) - V^{a_i}_{i,0}(X_{u}, S_{u}) \big) \,
I \big( X_u = S_u, \Xi^1_u = 0 \big) \, dN^{i,1}_u
\end{align}
is a local martingale under $(\GG^1, \PP)$, because of the continuous stochastic integral component
with respect to the standard Brownian motion ${\tilde B}^1 = ({\tilde B}^1_t)_{t \ge 0}$ is a local
martingale, where the process $N^{i,1} = (N^{i,1}_t)_{t \ge 0}$ is given by:
        \begin{align}
        \label{N5a}
        &N^{i,1}_t = \Xi^1_t - \int_0^t I \big( \Xi^1_u = 0 \big) \, \frac{\alpha}{S_u} \, {dS_u}
        \equiv \Xi^1_t - \int_0^{t \wedge \theta} \frac{\alpha}{S_u} \, {dS_u}
        \end{align}
        is a discontinuous uniformly integrable martingale under $(\GG^1, \PP)$,
        for each $0 < x \le s$ and every $i = 1, 2, 3$ and $j = 0, 1$ fixed.

        Note that, since the time spent by the process $(X, S)$ at the boundary surface
        $\partial C_{i,1,j} = \{ (x, s) \in E_{1} \, | \, x = a_{i,j}(s) \}$,
        for every $i = 1, 2, 3$ and $j = 0, 1$,
        as well as at the diagonal $d_{1} = \{ (x, s) \in E_{1} \, | \, x = s \}$
        is of the Lebesgue measure zero (cf., e.g. \cite[Chapter~II, Section~1]{BS}),
        the indicators in the second line of the formula in (\ref{rho4c}) as well as
        in the expression of (\ref{N5}) can be ignored.
        Moreover, since the component $S$ increases only when the process $(X, S)$ is located on the
        diagonal $d_{1}$, the indicator in the third line of (\ref{rho4c}) can also be set equal to one.
        Observe that the integral in the third line of (\ref{rho4c}) will actually be compensated
        accordingly, due to the fact that the candidate value function $V^{a_i}_{i,0}(x, s)$ satisfies
        the normal-reflection condition in the (\ref{NR1}) at the diagonal $d_{1}$,
        for every $i = 1, 2, 3$.

        For every $i = 1, 2, 3$ and $j = 0, 1$, it follows from straightforward calculations and the arguments of the previous section
        that the function $V^{a_i}_{i,j}(x, s)$ satisfies the second-order ordinary differential
        equation in (\ref{LV1}), which together with  (\ref{CF})-(\ref{SF})
        and (\ref{VD}) as well as the fact that the inequality in (\ref{LVD1}) holds imply that the inequality
        $(\LL V^{a_i}_{i,j} - r V^{a_i}_{i,j})(x, s) \le 0$ is satisfied with $H_{i,1,j}(x, s)$ given in (\ref{H5b2}),
        for all $0 < x < s$ such that $x \neq a_{i,j}(s)$, as well as every $i = 1, 2, 3$ and $j = 0, 1$.
        Moreover, we observe directly from the expressions in (\ref{Vg32a}) and (\ref{H'4a}) and (\ref{H''4a})
        with (\ref{CDg32a}) that the function $V^{a_i}_{i,j}(x, s)$ is convex and decreases to zero,
        because its first-order partial derivative $\partial_x V^{a_i}_{i,j}(x, s)$ is negative and increases to zero,
        while its second-order partial derivative $\partial_{xx} V^{a_i}_{i,j}(x, s)$ is positive, on the interval
        $a_{i,j}(s) < x \le s$, under $\alpha < 0$, for each $s > {\underline s}_{i,j}$.
        Thus, we may conclude that the inequality in (\ref{VC}) holds, which together with the
        conditions in (\ref{CF})-(\ref{SF}) and (\ref{VD}) imply that the inequality $V^{a_i}_{i,j}(x, s) \ge 0$
        is satisfied, for all $0 < x \le s$ as well as every $i = 1, 2, 3$ and $j = 0, 1$.
        Let $(\kappa_{i,n})_{n \in \NN}$ be the localising sequence
        of stopping times for the process $M^{i,1}$ from (\ref{N5}) such that
        $\kappa_{i,n} = \inf\{t \ge 0 \, | \, |M^{i,1}_t| \ge n \}$, for each
        $n \in \NN$ and every $i = 1, 2, 3$.
        It therefore follows from the expression in (\ref{rho4c}) that the inequalities:
        \begin{align}
        \label{rho4e}
        e^{- r (\tau \wedge \kappa_{i,n})} \,
        G_i(X_{\tau \wedge \kappa_{i,n}}, S_{\tau \wedge \kappa_{i,n}})
        &\le e^{- r (\tau \wedge \kappa_{i,n})} \, V^{a_i}_{i,\Xi^1_{\tau \wedge \kappa_{i,n}}}
        (X_{\tau \wedge \kappa_{i,n}}, S_{\tau \wedge \kappa_{i,n}}) \\
        \notag
        &\le V^{a_i}_{i,j}(x, s) + M^{i,1}_{\tau \wedge \kappa_{i,n}}
        \end{align}
        hold, for any stopping time $\tau$ of the process $(X, S, \Xi^1)$
        and each $n \in \NN$ as well as every $i = 1, 2, 3$ and $j = 0, 1$.
        Then, taking the expectation with respect to $\PP_{x, s, j}$
        in (\ref{rho4e}), by means of Doob's optional sampling theorem, we get:
        \begin{align}
        \label{rho4e2}
        \EE_{x, s, j} \big[ e^{- r (\tau \wedge \kappa_{i,n})} \,
        G_i(X_{\tau \wedge \kappa_{i,n}}, S_{\tau \wedge \kappa_{i,n}}) \big]
        &\le \EE_{x, s, j} \big[ e^{- r (\tau \wedge \kappa_{i,n})} \,
        V^{a_i}_{i,\Xi^1_{\tau \wedge \kappa_{i,n}}}(X_{\tau \wedge \kappa_{i,n}}, S_{\tau \wedge \nu_{i,k}}) \big]
        \\
        \notag
        &\le V^{a_i}_{i,j}(x, s) + \EE_{x, s, j} \big[ M^{i,1}_{\tau \wedge \kappa_{i,n}} \big] = V^{a_i}_{i,j}(x, s)
        \end{align}
        for all $0 < x \le s$ and each $n \in \NN$ as well as every $i = 1, 2, 3$ and $j = 0, 1$.
        Hence, letting $n$ go to infinity and using Fatou's lemma, we obtain from the expressions
        in (\ref{rho4e2}) that the inequalities:
        \begin{align}
        \label{rho4e3}
        \EE_{x, s, j} \big[ e^{- r \tau} \, G_i(X_{\tau}, S_{\tau}) \big]
        &\le \EE_{x, s, j} \big[ e^{- r \tau} \, V^{a_i}_{i,\Xi^1_{\tau}}(X_{\tau}, S_{\tau}) \big]
        \le V^{a_i}_{i,j}(x, s)
        \end{align}
        are satisfied, for any stopping time $\tau$ and all $0 < x \le s$
        such that $s > {\underline s}_{i,j}$ as well as every $i = 1, 2, 3$ and $j = 0, 1$.
        Thus, taking the supremum over all stopping times $\tau$ and then the infimum
        over all candidate boundaries $a_{i,j}$ in the expressions of (\ref{rho4e3}),
        we may therefore conclude that the inequalities:
        \begin{align}
        \label{rho4e4}
        &\sup_{\tau} \EE_{x, s, j} \big[ e^{- r \tau} \, G_i(X_{\tau}, S_{\tau}) \big]
        \le \inf_{a_{i}} V^{a_i}_{i,j}(x, s) = V^{a^*_i}_{i,j}(x, s)
        \end{align}
        hold, for all $0 < x \le s$, where $a^*_{i,j}(s)$ is the maximal solution to
        either the ordinary differential equation in (\ref{dg32a}) or the arithmetic
        equation in (\ref{hatUeq1}) as well as satisfying
        the inequality $a^*_{i,j}(s) < g^*_i(s) \wedge {\overline a}_{i,j}(s) \wedge s$,
        for all $s > {\underline s}_{i,j}$ as well as every $i = 1, 2, 3$ and $j = 0, 1$.

        By using the fact that the function $V^{a_i}_{i,j}(x, s)$ is (strictly)
        decreasing in the value $a_{i,j}(s)$, for each $s > {\underline s}_{i,j}$ fixed,
        we see that the infimum in (\ref{rho4e4})
        is attained over any sequence of solutions $(a_{i,j,m}(s))_{m \in \NN}$
        to (\ref{dg32a}) satisfying the inequality
        $a_{i,j,m}(s) < g^*_i(s) \wedge {\overline a}_{i,j}(s) \wedge s$, for all $s > {\underline s}_{i,j}$,
        and each $m \in \NN$, such that
        $a_{i,j,m}(s) \uparrow a^*_{i,j}(s)$ as $m \to \infty$,
        for each $s > {\underline s}_{i,j}$ fixed as well as every $i = 1, 2, 3$ and $j = 0, 1$.
It follows from the (local) uniqueness of the solutions to the first-order (nonlinear) ordinary
differential equation in (\ref{dg32a}) that no distinct solutions intersect, so that the sequence
$(a_{i,j,m}(s))_{m \in \NN}$ is decreasing and the limit
$a^*_{i,j}(s) = \lim_{m \to \infty} a_{i,j,m}(s)$ exists, for each $s > {\underline s}_{i,j}$ fixed.
Since the inequalities in (\ref{rho4e3}) hold for $a^*_{i,j}(s)$ too, we see that
the expression in (\ref{rho4e4}) holds, for $a^*_{i,j}(s)$ and all $0 < x \le s$
as well as every $i = 1, 2, 3$ and $j = 0, 1$.
We also note from the inequality in (\ref{rho4e2}) that the function $V^{a_i}_{i,j}(x, s)$
is superharmonic for the Markov process $(X, S)$, for all $0 < x \le s$
as well as every $i = 1, 2, 3$ and $j = 0, 1$.
Hence, taking into account the facts that $V^{a_i}_{i,j}(x, s)$ is decreasing
in $a_{i,j}(s) < g^*_i(s) \wedge {\overline a}_{i,j}(s) \wedge s$, for all
$s > {\underline s}_{i,j}$ as well as every $i = 1, 2, 3$ and $j = 0, 1$, while the inequality
$V^{a_i}_{i,j}(x, s) \ge 0$ holds, for all $0 < x \le s$, we observe that the selection
of the maximal solution $a^*_{i,j}(s)$, which stays strictly below the diagonal
$d_{1} = \{ (x, s) \in E_{1} \, | \, x = s \}$ and the curve
$\{ (x, s) \in E_{1} \, | \, x = {\overline a}_{i,j}(s) \}$, for every $i = 1, 2, 3$
and $j = 0, 1$, is equivalent to the implementation of the superharmonic
characterisation of the value function as the smallest superharmonic function dominating
the payoff function (cf. \cite{Pmax} or \cite[Chapter~I and Chapter~V, Section~17]{PSbook}).

In order to prove the fact that the boundary $a^*_{i,j}(s)$ is optimal, we consider the sequence
of stopping times $\tau_{i,m}$, for $m \in \NN$, defined as in the left-hand part of (\ref{tau*}) with
$a_{i,j,m}(s)$ instead of $a^*_{i,j}(s)$, where $a_{i,j,m}(s)$ is a solution to the first-order
ordinary differential equation in (\ref{dg32a}) and such that $a_{i,j,m}(s) \uparrow a^*_{i,j}(s)$
as $m \to \infty$, for each $s > {\underline s}_{i,j}$ as well as every $i = 1, 2, 3$ and $j = 0, 1$.
Then, by virtue of the fact that the function $V^{a_{i,m}}_{i,j}(x, s)$ from the left-hand side
of the expression in (\ref{rho3c}) associated with the boundary $a_{i,j,m}(s)$
satisfies equation (\ref{LV1}) and  (\ref{CF}),
and taking into account the structure of $\tau^*_i$ in (\ref{tau*}),
it follows from the expression which is equivalent to the one
in (\ref{rho4c}) that the equalities:
        \begin{align}
        \label{rho4g}
        e^{- r (\tau_{i,m} \wedge \kappa_{i,n})} \, G_i(X_{\tau_{i,m} \wedge \kappa_{i,n}},
        S_{\tau_{i,m} \wedge \kappa_{i,n}})
        &= e^{- r (\tau_{i,m} \wedge \kappa_{i,n})} \,
        V^{a_{i,m}}_{i,\Xi^1_{\tau_{i,m} \wedge \kappa_{i,n}}}(X_{\tau_{i,m} \wedge \kappa_{i,n}},
        S_{\tau_{i,m} \wedge \kappa_{i,n}}) \\
        \notag
        &= V^{a_{i,m}}_{i,j}(x, s) + M^{i,1}_{\tau_{i,m} \wedge \kappa_{i,n}}
        \end{align}
        hold, for all $0 < x \le s$ such that $s > {\underline s}_{i,j}$
        and each $n, m \in \NN$ as well as every $i = 1, 2, 3$ and $j = 0, 1$.
        Observe that, by virtue of the arguments from \cite[pages~635-636]{SS1}, the property:
        \begin{align}
        \label{ESS4}
        &\EE_{x, s, j} \Big[ \sup_{t \ge 0} e^{- r (\tau^*_i \wedge t)} \,
        G_i(X_{\tau^*_i \wedge t}, S_{\tau^*_i \wedge t}) \Big] < \infty
        \end{align}
        holds, for all $0 < x \le s$ as well as every $i = 1, 2, 3$ and $j = 0, 1$.
        Hence, letting $m$ and $n$ go to infinity and using the condition of (\ref{CF})
        as well as the property $\tau_{i,m} \uparrow \tau^*_i$ ($\PP_{x, s, j}$-a.s.) as $m \to \infty$,
        we can apply the Lebesgue dominated convergence theorem to the appropriate (diagonal)
        subsequence in the expression of (\ref{rho4g}) to obtain the equality:
        \begin{align}
        \label{rho4i}
        &\EE_{x, s, j} \big[ e^{- r \tau^*_i} \, G_i(X_{\tau^*_i}, S_{\tau^*_i}) \big] = V^{a^*_i}_{i,j}(x, s)
        \end{align}
        for all $0 < x \le s$ such that $s > {\underline s}_{i,j}$ as well as every $i = 1, 2, 3$ and $j = 0, 1$,
        which together with the inequalities in (\ref{rho4e4}) directly implies the desired assertion. $\square$
\end{proof}

\section{\dot Appendix}

In this section, we refer the closed-form solutions of the perpetual American standard and
lookback option problems in (\ref{VU7a1}) which are known from \cite[Chapter~VIII, Section~2a]{FM},
\cite{SS1}, \cite{Jesper}, \cite{GuoShepp}, \cite{GapSPL20} and \cite[Section~5]{Gapmax24} for completeness.
Although, by means of the change-of-measure arguments from \cite{SS2} and \cite{Grus},
the problems of (\ref{VU7a1}) under $i = 2$ can be reduced to the appropriate optimal
stopping problems for one-dimensional Markov processes $S/X = (S_t/X_t)_{t \ge 0}$
and $Q/X = (Q_t/X_t)_{t \ge 0}$, we present their solutions as of the
two-dimensional optimal stopping problems.


\subsection{The perpetual American standard and lookback put and call options.}
For the original perpetual American option pricing problems in (\ref{VU5c1}),
let us now recall explicit expressions for the value functions and optimal
exercise boundaries of the auxiliary option pricing problems with the
payoff functions $G_i(x, s)$ and $F_i(x, q)$ given after (\ref{VU5a1}) above
as solutions to the optimal stopping problems:
\begin{equation}
\label{VU7a1}
{\bar V}_{i}(x, s) = \sup_{{\tau'} \in \FF} \EE_{x, s}
\big[ e^{- r \tau'} \, G_i(X_{\tau'}, S_{\tau'}) \big]
\quad \text{and} \quad
{\bar U}_{i}(x, q) = \sup_{{\zeta'} \in \FF} \EE_{x, q}
\big[ e^{- r \zeta'} \, F_i(X_{\zeta'}, Q_{\zeta'}) \big]
\end{equation}
where the suprema are taken over all stopping times ${\tau'}$ and ${\zeta'}$
with respect to the natural filtration $\FF$ of the process $X$, for every $i = 1, 2, 3$.
Here, $\EE_{x, s}$ and $\EE_{x, q}$
denote the expectations with respect to the probability measures $\PP_{x, s}$ and $\PP_{x, q}$
under which the two-dimensional Markov processes $(X, S)$ and $(X, Q)$ given by
(\ref{X4a})-(\ref{dX4a}) and (\ref{X4b})-(\ref{dX4b}) with (\ref{SQ4})
start at $(x, s)$ such that $0 < x \le s$ and $(x, q)$ such that $0 < q \le x$, respectively.
It can be shown by means of the same arguments as in Theorem~\ref{lem21} above that
the optimal exercise times have the form:
       \begin{equation}
       \label{tau7*}
       {{\bar \tau}_i'} = \inf \big\{ t \ge 0 \; \big| \; X_t \le g^*_i(S_t) \big\}
       \quad \text{and} \quad
       {{\bar \zeta}_i'} = \inf \big\{ t \ge 0 \; \big| \; X_t \ge h^*_i(Q_t) \big\}
       \end{equation}
       for some boundaries $0 < g^*_1 < L_1$ and $h^*_1 > K_1$ as well as
       $0 < g^*_i(s) < s$ and $h^*_i(q) > q$, for $i = 2, 3$, to be determined.
It follows from the results of general theory mentioned above that the continuation
regions $C_{i,k}'$ and $D_{i,k}'$, for $k = 1, 2$, for the optimal stopping problems
of (\ref{VU7a1}) have the form:
\begin{equation}
\label{C75}
C_{i,1}' = \big\{ (x, s) \in E_1 \; \big| \; {\bar V}_{i,1}(x, s) > G_i(x, s) \big\},
\;\;
C_{i,2}' = \big\{ (x, q) \in E_2 \; \big| \; {\bar U}_{i,1}(x, q) > F_i(x, q) \big\}
\end{equation}
and
\begin{equation}
\label{D75}
D_{i,1}' = \big\{ (x, s) \in E_1 \; \big| \; {\bar V}_{i,1}(x, s) = G_i(x, s) \big\},
\;\;
D_{i,2}' = \big\{ (x, q) \in E_2 \; \big| \; {\bar U}_{i,1}(x, q) = F_i(x, q) \big\}
\end{equation}
for every $i = 1, 2, 3$, respectively.



\subsection{The candidate value functions and exercise boundaries.}
It is well known since \cite{McKean} (cf., e.g. \cite[Chapter~VIII, Section~2a]{FM})
that the value functions ${\overline V}_1(x, s) \equiv {\overline V}_{1}(x)$ and
${\overline U}_{1}(x, q) \equiv {\overline U}_{1}(x)$ in (\ref{VU7a1}) admit the representations:
\begin{equation}
\label{Upc1}
{\overline V}_{1}(x; g_{1}) =
- \frac{g_{1}}{\gamma_2} \, \Big( \frac{x}{g_{1}} \Big)^{\gamma_2}
\quad \text{and} \quad
{\overline U}_{1}(x; h_{1}) =
\frac{h_{1}}{\gamma_1} \, \Big( \frac{x}{h_{1}} \Big)^{\gamma_1}
\end{equation}
while the optimal exercise boundaries take the form:
       \begin{equation}
       \label{gh1}
       g^*_{1} = \frac{\gamma_2 L_1}{\gamma_2 - 1}
       \quad \text{and} \quad
       h^*_{1} = \frac{\gamma_1 K_1}{\gamma_1 - 1}
       \end{equation}
       where $\gamma_l$, for $l = 1, 2$, are given by:
       \begin{equation}
       \label{gamma12}
       \gamma_l = \frac{1}{2}-\frac{r-\delta}{\sigma^2} - (-1)^l
       \sqrt{\bigg( \frac{1}{2} - \frac{r-\delta}{\sigma^2} \bigg)^2
       + \frac{2r}{\sigma^2}}
       \end{equation}
       so that $\gamma_2 < 0 < 1 < \gamma_1$ holds.

       It follows from the results of \cite{Jesper} and \cite{GuoShepp}
       that the candidate value functions in the left-hand problem,
       for $i = 2, 3$, are given by:
       \begin{align}
       \label{Vi71}
       {\overline V}_i(x, s; g_i(s)) =
       {\overline C}_{i, 1}(s; g_i(s)) \, x^{\gamma_1}
       + {\overline C}_{i, 2}(s; g_i(s)) \, x^{\gamma_2}
       \end{align}
       with
       \begin{align}
       \label{Ci71}
       {\overline C}_{i, j}(s; g_i(s))
       &= \frac{\gamma_{3-j} G_i(g_i(s), s) - g_i(s) \partial_x G_i(g_i(s), s)}
       {(\gamma_{3-j} - \gamma_j)(g^{\gamma_j}_i(s))}
       \end{align}
       for all $0 < g_i(s) < x \le s$ such that $s > {\underline s}_{i,0}$,
       with the optimal exercise boundary $g^*_i(s)$ being the maximal solution to the
       first-order nonlinear ordinary differential equation:
       \begin{align}
       \label{h71}
       &g_{i}'(s) =
       - \frac{\partial_s {\bar C}_{i,1}(s; g_i(s)) s^{\gamma_1}
       + \partial_s {\bar C}_{i,2}(s; g_i(s)) s^{\gamma_2}}
       {\partial_{g_i} {\bar C}_{i,1}(s; g_i(s)) s^{\gamma_1} +
       \partial_{g_i} {\bar C}_{i,2}(s; g_i(s)) s^{\gamma_2}}
       \end{align}
       such that $0 < g_{i}(s) < s$,
       for $s > {\underline s}_{i,0}$ and $i = 2, 3$.
       Moreover, it follows from the results of \cite{GapSPL20} and
       \cite[Section~5]{Gapmax24} that the candidate solution to the right-hand problem,
       for $i = 2, 3$, takes the form:
       \begin{align}
       \label{Vi72}
       {\overline U}_i(x, q; h_i(q))
       &= {\overline D}_{i, 1}(q; h_i(q)) \, x^{\gamma_1}
       + {\overline D}_{i, 2}(q; h_i(q)) \, x^{\gamma_2}
       \end{align}
       with
       \begin{align}
       \label{Ci72}
       {\overline D}_{i, j}(q; h_i(q))
       &= \frac{\gamma_{3-j} G_i(h_i(q), q) - h_i(q) \partial_x G_i(h_i(q), q)}
       {(\gamma_{3-j} - \gamma_j)(h^{\gamma_j}_i(q))}
       \end{align}
       for $0 < q \le x < h^*_i(q)$ such that $q < {\overline q}_{i,0}$,
       with the optimal exercise boundary $h^*_i(q)$ being the minimal
       solution to the first-order nonlinear ordinary differential equation:
       \begin{align}
       \label{h72}
       &h_{i}'(q) =
       - \frac{\partial_q {\bar D}_{i,1}(q; h_i(q)) q^{\gamma_1}
       + \partial_q {\bar D}_{i,2}(q; h_i(q)) q^{\gamma_2}}
       {\partial_{h_i} {\bar D}_{i,1}(q; h_i(q)) q^{\gamma_1} +
       \partial_{h_i} {\bar D}_{i,2}(q; h_i(q)) q^{\gamma_2}}
       \end{align}
       such that $h_{i}(q) > q$, for $q < {\overline q}_{i,0}$.

       In particular, for $i = 2$, it is seen from the expressions in (\ref{h71})
       and (\ref{h72}) that $g^*_{2}(s) \equiv {\underline \lambda} s$,
       for $s > 0$, and $h^*_2(q) \equiv {\overline \nu} q$,
       for $q > 0$, where the numbers $0 < {\underline \lambda} < 1$
       and ${\overline \nu} > 1$ provide the unique roots of
       the power arithmetic equations:
       \begin{align}
       \label{h71la}
       &\lambda^{\gamma_{1} - \gamma_{2}} =
       \frac{(\gamma_1 - 1) (\gamma_2 (1 - L_2 \lambda) + L_2 \lambda)}
       {(\gamma_2 - 1) (\gamma_1 (1 - L_2 \lambda) + L_2 \lambda)}
       \quad \text{and} \quad
       \nu^{\gamma_{1} - \gamma_{2}} =
       \frac{(\gamma_1 - 1) (\gamma_2 (1 - K_2 \nu) + K_2 \nu)}
       {(\gamma_2 - 1) (\gamma_1 (1 - K_2 \nu) + K_2 \nu)}
       \end{align}
       on the interval $(0, 1)$ and $(1, \infty)$, respectively.
       Furthermore, for $i = 3$, by means of straightforward computations,
       it follows that the first-order nonlinear ordinary differential
       equations in (\ref{h71}) and (\ref{h72}) take the form:
       \begin{align}
       \label{h73}
       g_{3}'(s) &= \frac{g_{3}(s)}{s - L_3} \, \frac{\gamma_2 (s/g^{\gamma_1}_{3}(s)) -
       \gamma_1 (s/g^{\gamma_2}_{3}(s))}{\gamma_1 \gamma_2 ((s/g^{\gamma_1}_{3}(s))
       - (s/g^{\gamma_2}_{3}(s)))}
       \end{align}
       for all $s > {\underline s}_{3,0} \equiv L_3$, and
       \begin{align}
       \label{h74}
       h_{3}'(q) &= \frac{h_3(q)}{K_3 - q} \, \frac{\gamma_2 (q/h^{\gamma_1}_{3}(q)) -
       \gamma_1 (q/h^{\gamma_2}_{3}(q))} {\gamma_1 \gamma_2 ((q/h^{\gamma_1}_{3}(q))
       -(q/h^{\gamma_2}_{3}(q)))}
       \end{align}
       for all $q < {\overline q}_{3,0} \equiv K_3$, respectively.

\subsection{The results.}
Summarising the facts shown above, we state the following result which follows from the results
of \cite[Chapter~VIII, Section~2a]{FM}, \cite{Jesper}, \cite{GuoShepp}, \cite{GapSPL20} and
\cite[Section~5]{Gapmax24}.

      \begin{proposition}
      \label{CorA1}
      Let the processes $(X, S)$ and $(X, Q)$ be given by (\ref{X4})-(\ref{dX4}) and (\ref{SQ4}) with $r > 0$, $\delta > 0$ and $\sigma > 0$. Then, the value functions of the perpetual American standard and lookback put and call options from (\ref{VU7a1}) admit the expressions:
      \begin{equation}
      \label{W*71a}
      {\bar V}^*_i(x, s) =
      \begin{cases}
      {\bar V}_i(x, s; g^*_i(s)), & \text{if} \quad g^*_i(s) < x \le s, \\
      G_i(x, s), & \text{if} \quad 0 < x \le g^*_i(s),
      \end{cases}
      \end{equation}
      and
      \begin{equation}
      \label{W*71b}
      {\bar U}^*_i(x, q) =
      \begin{cases}
      {\bar U}_i(x, q; h^*_i(q)), & \text{if} \quad 0 < q \le x < h^*_i(q), \\
      F_i(x, q), & \text{if} \quad x \ge h^*_i(q),
      \end{cases}
      \end{equation}
      while the optimal exercise times have the form of (\ref{tau7*}) above,
      where the candidate value functions and exercise boundaries are specified as follows:

(i) the function ${\bar V}_1(x, s; g^*_1(s)) \equiv {\bar V}_1(x; g^*_1)$
is given by (\ref{Upc1}) with $g^*_1$ from (\ref{gh1}), the functions ${\bar V}_i(x, s; g^*_i(s))$, for $i = 2, 3$, is given by (\ref{Vi71})-(\ref{Ci71}) with $g^*_{2}(s) \equiv {\underline \lambda} s$ and ${\underline \lambda}$ being a unique solution to
the arithmetic equation in (\ref{h71la}) and $g^*_{3}(s)$ being a maximal solution to the ordinary differential equation in (\ref{h73});

(ii) the function ${\bar U}_1(x, q; h^*_1(q)) \equiv {\bar U}_1(x; h^*_1)$
is given by (\ref{Upc1}) with $h^*_1$ from (\ref{gh1}), the functions ${\bar U}_i(x, q; h^*_i(q))$, for $i = 2, 3$, is given by (\ref{Vi72})-(\ref{Ci72}) with $h^*_{2}(q) \equiv {\overline \nu} q$ and ${\overline \nu}$ being a unique solution of the arithmetic
equation in (\ref{h71la}) and $h^*_{3}(q)$ being a maximal solution to the ordinary differential equation in (\ref{h74}).
\end{proposition}




\end{document}